%% file: main.tex
\newif\ifshort
\documentclass[a4paper,UKenglish,cleveref, autoref, USenglish, thm-restate,nosubfigcap,nolineno]{socg-lipics-v2021}
\hideLIPIcs
\usepackage{amsmath,amsfonts,amssymb,amsthm} 
\usepackage[utf8]{inputenc}
\usepackage{color}
\usepackage{url}
\usepackage{framed}
\usepackage{xcolor}
\usepackage{todonotes}
\usepackage{amssymb}
\usepackage{xspace}
\usepackage{nicefrac}
 \usepackage{xifthen}
 \usepackage{hyperref}
\usepackage{xspace} 
\usepackage{framed}
\usepackage{thm-restate}
\usepackage{mdframed}
\usepackage{tikz}
 \usepackage{forest}
\usepackage{subcaption}
\usetikzlibrary{calc,math,patterns,shapes,backgrounds}
\tikzstyle{path} = [color=black,opacity=.30,line cap=round, line join=round, line width=10pt]
 \usepackage{graphicx,color}
\usepackage{boxedminipage}
\usepackage{framed}
\usepackage{mathtools}
\usepackage{thmtools}
\usepackage{amsmath,amsfonts}
\usepackage[ruled,vlined]{algorithm2e} 
\usepackage{footnote}
\usepackage{nicefrac}
\usepackage{todonotes}
\usepackage{footnote}
 \usepackage{xifthen}
 \usepackage{tabularx}
\usepackage{subcaption}

\crefname{redrule}{Reduction Rule}{Reduction Rule}
\newtheorem{dtredrule}{Reduction Rule}
\newtheorem*{dtredrule*}{Reduction Rule}
\crefname{dtredrule*}{Reduction Rule}{Reduction Rule}
\crefname{dtredrule}{Reduction Rule}{Reduction Rule}
\crefname{algocfline}{Algorithm}{Algorithm}

\crefformat{framedbox}{branching subroutine}
\Crefformat{framedbox}{Branching Subroutine}

\newcommand{\labelframedbox}[1]{%
  \phantomsection
  \def\@currentlabelname{Framed Box}%
  \label[#1]{#1}%
}

\newcommand{\cO}{\mathcal{O}}

\newcommand{\mcm}[3]{\newcommand{#1}[#2]{{\ensuremath{#3}}}} 
\mcm{\Nbb}{0}{\mathbb{N}}
\mcm{\Zbb}{0}{\mathbb{Z}}
\mcm{\Rbb}{0}{\mathbb{R}}
\mcm{\Cbb}{0}{\mathbb{C}}
\mcm{\Qbb}{0}{\mathbb{Q}}
\mcm{\Acal}{0}{\cal A}
\mcm{\Bcal}{0}{\cal B}
\mcm{\Ccal}{0}{\cal C}
\mcm{\Dcal}{0}{\cal D}
\mcm{\Ecal}{0}{\cal E}
\mcm{\Fcal}{0}{\cal F}
\mcm{\Gcal}{0}{\cal G}
\mcm{\Hcal}{0}{\cal H}
\mcm{\Ical}{0}{\cal I}
\mcm{\Jcal}{0}{\cal J}
\mcm{\Kcal}{0}{\cal K}
\mcm{\Lcal}{0}{\cal L}
\mcm{\Mcal}{0}{\cal M}
\mcm{\Ncal}{0}{\cal N}
\mcm{\Ocal}{0}{{\cal O}}
\mcm{\Pcal}{0}{{\cal P}}
\mcm{\Qcal}{0}{{\cal Q}}
\mcm{\Rcal}{0}{{\cal R}}
\mcm{\Scal}{0}{{\cal S}}
\mcm{\Tcal}{0}{{\cal T}}
\mcm{\Ucal}{0}{{\cal U}}
\mcm{\Vcal}{0}{{\cal V}}
\mcm{\Wcal}{0}{{\cal W}}
\mcm{\Xcal}{0}{{\cal X}}
\mcm{\Ycal}{0}{{\cal Y}}
\mcm{\Zcal}{0}{{\cal Z}}

\DeclareMathOperator{\operatorClassFPT}{FPT\xspace}
\newcommand{\classFPT}{\ensuremath{\operatorClassFPT}\xspace}

\newcommand{\yes}{{yes}}
\newcommand{\no}{{no}}
\newcommand{\yesinstance}{\yes-instance\xspace}
\newcommand{\noinstance}{\no-instance\xspace}
\newcommand{\yesinstances}{\yesinstances\xspace}
\newcommand{\noinstances}{\no-instances\xspace}

\newcommand{\normal}{normal\xspace}
\newcommand{\fixed}{fixed\xspace}

\newcommand{\unfix}{unfit\xspace}
\newcommand{\Normal}{Normal\xspace}
\newcommand{\Fixed}{Fixed\xspace}

\newcommand{\Unfix}{Unfit\xspace}
\newcommand{\NORM}{\mathcal{N}}
\newcommand{\FIXED}{\mathcal{F}}

\newcommand{\UNFIX}{\mathcal{U}}
\newcommand{\treshold}{\tau}
\newcommand{\Treshold}{\mathcal{T}}
\newcommand{\Dvalues}{\mathcal{D}}
\newcommand{\wrt}{w.r.t. }

\newcommand{\dist}{\mathsf{dist}}

\SetKwInput{kwInit}{Initialization} 

\newlength{\RoundedBoxWidth}
\newsavebox{\GrayRoundedBox}
\newenvironment{GrayBox}[1]%
   {\setlength{\RoundedBoxWidth}{.93\textwidth}
    \def\boxheading{#1}
    \begin{lrbox}{\GrayRoundedBox}
       \begin{minipage}{\RoundedBoxWidth}}%
   {   \end{minipage}
    \end{lrbox}
    \begin{center}
    \begin{tikzpicture}%
       \node(Text)[draw=black!20,fill=white,rounded corners,%
             inner sep=2ex,text width=\RoundedBoxWidth]%
             {\usebox{\GrayRoundedBox}};
        \coordinate(x) at (current bounding box.north west);
        \node [draw=white,rectangle,inner sep=3pt,anchor=north west,fill=white] 
        at ($(x)+(6pt,.75em)$) {\boxheading};
    \end{tikzpicture}
    \end{center}}     

\newenvironment{defproblemx}[2][]{\noindent\ignorespaces%
                                \FrameSep=6pt%
                                \parindent=0pt%
                \vspace*{-1.5em}
                \ifthenelse{\isempty{#1}}{%
                  \begin{GrayBox}{\textsc{#2}}%
                }{%
                  \begin{GrayBox}{\textsc{#2}  parameterized by~{#1}}%
                }
                \begin{tabular*}{\textwidth}{@{\hspace{.1em}} >{\itshape} p{1.8cm} p{0.8\textwidth} @{}}%
            }{
                \end{tabular*}%
                \end{GrayBox}%
                \ignorespacesafterend
            }  

\newcommand{\defparproblema}[4]{
  \begin{defproblemx}[#3]{#1}
    Input:  & #2 \\
    Task: & #4
  \end{defproblemx}
}%

\newcommand{\uvd}{\textsc{Ultrametric Violation Distance}\xspace}

\newcommand{\cuvd}{\textsc{Constrained Ultrametric Violation Distance}\xspace}
\newcommand{\Cuvd}{\textsc{CUVD}\xspace}
\ifshort
    \title{Ultrametric Violation Distance: Polynomial Kernel and FPT Algorithm}
\else
    \title{Ultrametric Violation Distance: Polynomial Kernel and FPT Algorithm}
\fi

\author{Fedor V. Fomin}{University of Bergen, Bergen, Norway}
{Fedor.Fomin@uib.no}{https://orcid.org/0000-0003-1955-4612}
{Supported by the European Research Council (ERC) under the European Union's Horizon 2020 research and innovation programme (NewPC grant agreement No.  101199930).\\\begin{minipage}{0.2\textwidth}\includegraphics[width=0.9\textwidth]{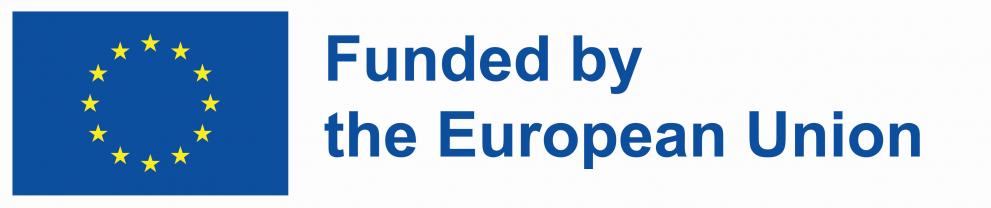}\end{minipage} }

\author{Petr A. Golovach}{University of Bergen, Norway}
{petr.golovach@uib.no}{https://orcid.org/0000-0002-2619-2990}
{Supported by the Research Council of Norway under BWCA (grant no.~314528) and Extreme Algorithms  (grant no.~355137) projects} 

\author{Yash Hiren More}{University of Bergen, Norway}
{yash.h.more@uib.no}{https://orcid.org/0000-0002-8651-6686}
{Supported by the Trond Mohn forskningsstiftelse (grant no. TMS2023TMT01)}

\authorrunning{F.V.~Fomin, P.A.~Golovach, and Y.~More}

\Copyright{Fedor V. Fomin, Petr A. Golovach, and Yash More}



\date{}

\makeatletter
\def\ps@headings{%
  \def\@evenhead{\large\sffamily\bfseries
    \llap{\hbox to0.5\oddsidemargin{\thepage\hss}}\leftmark\hfil}%
  \def\@oddhead{\large\sffamily\bfseries
    \rightmark\hfil
    \rlap{\hbox to0.5\oddsidemargin{\hss\thepage}}}%
  \def\@oddfoot{\hfil
    \rlap{%
      \vtop{%
        \vskip10mm
        \colorbox{lipicsYellow}{%
          \@tempdima\evensidemargin
          \advance\@tempdima1in
          \advance\@tempdima\hoffset
          \hb@xt@\@tempdima{%
            \textcolor{lipicsGray}{\normalsize\sffamily
            \bfseries\quad
            \expandafter\textsolittle
            \expandafter{\@EventShortTitle}}%
            \strut\hss}}}}}%
  \let\@evenfoot\@empty
  \let\@mkboth\markboth
  \let\sectionmark\@gobble
  \let\subsectionmark\@gobble
}
\makeatother

\begin{document}
\maketitle

\begin{abstract}
In the \textsc{Ultrametric Violation Distance} problem, we are given a set of distances between $n$ points, and the goal is to modify the minimum number of distances so that the resulting set forms a valid ultrametric. In other words, the task is to fit an ultrametric to the given data, where the quality of the fit is measured by the $\ell_0$-norm of the error.

While variants of this problem under the $\ell_\infty$ and $\ell_1$-norms have been well studied, the complexity of \textsc{Ultrametric Violation Distance} under the $\ell_0$-norm remained largely unexplored until recently. This changed with the work of Cohen-Addad, Fan, Lee, and Mesmay [FOCS~2022], who introduced a constant-factor approximation algorithm. Significant further progress on approximation algorithms was made in subsequent work by Charikar and Gao [SODA~2024], and by An, Kao, Lee, and Lee [FOCS~2025].

In this paper, we initiate a systematic study of \textsc{Ultrametric Violation Distance} from the perspectives of kernelization and fixed-parameter tractability (FPT). By the work of Fan, Gilbert, Raichel, Sonthalia, and Van~Buskirk [SWAT~2020], the problem is known to be FPT when parameterized by the number of violated distances $k$. We show that the problem admits a kernel with $\mathcal{O}(k^2)$ points. Additionally, we present a single-exponential-time algorithm with running time $9^k \cdot n^{\mathcal{O}(1)}$, which is asymptotically tight.

\end{abstract}

\newpage

\section{Introduction}\label{sec:intro}

Hierarchical clustering is a method of cluster analysis that builds a hierarchy of clusters by initially treating each data point as its own cluster and successively merging the two closest clusters. This process continues until all points are merged into a single cluster or a predefined stopping criterion is reached.
Hierarchical clustering is a widely used technique, with applications ranging from gene expression analysis \cite{dhaeseleer2005how} and phylogenetic tree construction \cite{ailon2011fitting,kimes2017statistical} to image processing \cite{lee2005unsupervised}, social network analysis \cite{breiger1975algorithm}, and marketing \cite{kumar2020digital}.

Hierarchical clustering is closely related to the concept of ultrametrics in mathematics \cite{Carlsson2010}. An ultrametric is a special kind of distance measure that satisfies a strengthened version of the triangle inequality: for any three points, the two largest distances are equal. In other words, a metric space on a finite point-set $X$ with distance function $\dist$ is an \emph{ultrametric space} if for every distinct $x,y,z \in X$,
\[
\dist(x,z) \leq \max\{\dist(x,y), \dist(y,z)\}.
\]
This property naturally gives rise to the following hierarchical structure among data points.  
Let $T$ be a finite, rooted tree (a dendrogram) whose leaves are points of a finite point-set $X$.  
Let $w : V(T) \setminus X \to \mathbb{R}_{\ge 0}$ be a function that assigns non-negative weights to the internal vertices of $T$ such that the vertex weights are non-increasing along root-to-leaf paths. Then $X$ with the distance function $\dist_T$ defined on $x,y \in X$ as
\[
\dist_T(x,y) = w(\mathrm{LCA}(x,y)),
\]
where $\mathrm{LCA}(x,y)$ is the lowest common ancestor of $x$ and $y$ in $T$,  is ultrametric. See \Cref{fig:dendtree}.

\begin{figure}
  \centering
  \begin{subfigure}[b]{0.45\textwidth}
    \centering
    \begin{forest}
        for tree={
            grow=east,
            draw,
            edge path={\noexpand\path [draw] (!u.parent anchor) -- ++(5pt,0) |- (.child anchor);},
            parent anchor=east,
            child anchor=west,
            l=1cm,           
            s sep=0.3cm,     
            inner sep=2pt,
            if level=0{}{font=\small}
        }
          [{$3$}
            [{$2$}
              [{$1$} [A] [B] ]
              [{$1$} [C] [D] ]
            ]
            [{$2$}
              [{$1$} [E] [F] ]
              [{$1$} [G] [H] ]
            ]
          ]
    \end{forest}
    \caption{Dendrogram for 8-point ultrametric}  
  \end{subfigure}
  \hfill
  \begin{subfigure}[b]{0.45\textwidth}
    \centering
    \begin{tabular}{c|cccccccc}
       & A & B & C & D & E & F & G & H \\ \hline
      A & 0 & 1 & 2 & 2 & 3 & 3 & 3 & 3 \\
      B &  & 0 & 2 & 2 & 3 & 3 & 3 & 3 \\
      C &  &  & 0 & 1 & 3 & 3 & 3 & 3 \\
      D &  &  &  & 0 & 3 & 3 & 3 & 3 \\
      E &  &  &  &  & 0 & 1 & 2 & 2 \\
      F &  &  &  &  &  & 0 & 2 & 2 \\
      G &  &  &  &  &  &  & 0 & 1 \\
      H &  &  &  &  &  &  &  & 0 \\
    \end{tabular}
    \caption{Ultrametric distance table}
  \end{subfigure}
  \caption{An ultrametric on eight points \(A,\dots,H\):  
    pairs \(\{A,B\},\{C,D\},\{E,F\},\{G,H\}\) merge at height 1,  
    quartets \(\{A,B,C,D\}\) and \(\{E,F,G,H\}\) at height 2,  
    and the full set at height 3.}\label{fig:dendtree}.
    (Note that only the upper triangle is shown; the matrix is symmetric.)
\end{figure}
 
The fundamental problem of fitting distances with ultrametrics, and more generally, with 
tree metrics, known as the \emph{numerical taxonomy problem}, has been a subject of interest since the 1960s~\cite{cavalli1967phylogenetic,sneath1962numerical}. In this problem, given a matrix of distances $D$, not necessarily a metric, one seeks the best fit of $D$ by an ultrametric matrix $D'$. Different formulations of the optimal fit for a given distance matrix $D$ lead to various objectives. The two classical ones are minimizing the sum of differences (i.e., minimizing the $\ell_1$-norm of matrix $D - D'$) and minimizing the maximum error (i.e., minimizing the $\ell_\infty$-norm).

In 2022, Cohen-Addad, Fan, Lee, and Mesmay~\cite{cohen2022fitting} initiated the study of minimizing the number of disagreements (i.e., minimizing the $\ell_0$-norm of $D - D'$) under the name \uvd. They provided a constant-factor approximation algorithm. This work initiated a surge of interest in approximation algorithms; see the section on Related Work.

In this paper, we study \uvd from the perspective of parameterized complexity. We focus on the ``natural'' parameterization by the number of violated distances $k$. The fact that \uvd is in \classFPT when parameterized by $k$ implicitly follows from the work of Fan, Gilbert, Raichel, Sonthalia, and Van~Buskirk~\cite{fan_et_al:LIPIcs.SWAT.2020.25}. In that work (Theorem~15 in~\cite{fan_et_al:LIPIcs.SWAT.2020.25}), they give a $k^{\mathcal{O}(k)} n^{\mathcal{O}(1)}$-time algorithm for the \textsc{Metric Violation Distance} problem. In this problem, one is allowed to change at most $k$ distances in order to transform a distance space into a metric. (Their algorithm also applies to a more general setting of metric repair on $\zeta$-chordal graphs.)  
It is known that \uvd can be reduced to an instance of \textsc{Metric Violation Distance} by tropicalizing the distances via the map $d \mapsto n^d$; see, for example, \cite[Corollary~5.16]{CohenAddadFLM25}. Combined with  ~\cite{fan_et_al:LIPIcs.SWAT.2020.25}, this immediately yields a $k^{\mathcal{O}(k)} n^{\mathcal{O}(1)}$-time algorithm for \uvd.

Following the established research pipeline in parameterized complexity, namely, first establishing fixed-parameter tractability, then seeking kernelization results and tighter parameter dependence, we are led to the following open questions. The first question 
 is on the existence of a polynomial kernel. 
 
 \begin{framed}
\noindent
\textbf{Question 1 (Kernelization).}  
Does \uvd admit a polynomial kernel when parameterized by the number of violated distances $k$?
\end{framed}

The second question is on the optimality of the exponential dependence in $k$ in the running time $k^{\mathcal{O}(k)} n^{\mathcal{O}(1)}$. 

\begin{framed}
\noindent
\textbf{Question 2 (Optimal parameter dependence).}  
What is the optimal dependence on $k$ in the running time of an FPT algorithm for \uvd? In particular, can the current $k^{\mathcal{O}(k)}$ dependence be improved?
\end{framed}

We note that it is highly unlikely that the exponential dependence on $k$ can be improved to a subexponential running time of the form $2^{o(k)} \cdot n^{\mathcal{O}(1)}$. This follows from established hardness results for \textsc{Cluster Editing}, which is a special case of \uvd. Since \textsc{Cluster Editing} does not admit a subexponential-time algorithm of the form $2^{o(k)} \cdot n^{\mathcal{O}(1)}$ unless the Exponential Time Hypothesis (ETH) fails~\cite{FominKPPV14,komusiewicz2012cluster}, the same limitation applies to \uvd. However, it still does not exclude the existence of a single-exponential running time 
$2^{O(k)} \cdot n^{\mathcal{O}(1)}$.

\subsection{Our results}
Before proceeding with describing our results, let us recall the well-known connection between \uvd and the \textsc{Cluster Editing} problem.  
A graph $G$ is a \emph{cluster graph} if every connected component of $G$ is a complete graph.  
Let $G$ be a graph, and let 
$F\subseteq \binom{V(G)}{2}$ be a set of unordered pairs of vertices
such that $G \triangle F$ is a cluster graph. 
Then $F$ is called a \emph{cluster editing set} for $G$.  
Here, $G \triangle F$ is the graph on $V(G)$ whose edge set is the symmetric difference between $E(G)$ and $F$. The task is to find a cluster editing set of size at most $k$.

Consider an ultrametric \( \Xcal = (X, \dist) \), where the function $\dist$ takes only two values, say $1$ or $2$. Construct an auxiliary graph $G$ on vertex set $X$, with edge set $xy \in E(G)$ if and only if $\dist(x,y) = 1$.  
It is straightforward to verify that $G$ is a cluster graph. Indeed, if it is not a cluster graph, it must contain an induced path $P_3 = xyz$ on three vertices. But then $\dist(x,y) = \dist(y,z) = 1$ and $\dist(x,z) = 2$, which contradicts the ultrametric property.
Thus, the \uvd{} problem, when the distance function takes only two values, is equivalent to finding the minimum cluster editing set for the auxiliary graph.

In this paper, we solve a more general version of the \uvd{} problem. 
Our problem is related to \textsc{Ultrametric Matrix Sandwich Problem} of \cite{farach1995robust}, which asks whether there exists an ultrametric satisfying given lower and upper bounds, but does not take an input distance function on $X$.

Formally, we are given two symmetric functions: a lower threshold 
 $\treshold_{\ell} \colon X \times X \to \mathbb{R}_{\geq 0}$ 
and an upper threshold 
$\treshold_{u} \colon X \times X \to \mathbb{R}_{\geq 0}$.

The meaning of these thresholds is that the final distance $\dist(x,y)$ must be within the interval $[\treshold_{\ell}(x,y), \treshold_{u}(x,y)]$.  
These thresholds arise naturally in our algorithm, so it is logical for us to consider this generalized version of the problem. \uvd is the special case of the constrained problem with $\treshold_{\ell}=0$ and $\treshold_{u}(x,y)=\max D$.

We define \cuvd (\Cuvd for short) as the following problem.
The input is an integer $k\geq 0$ and a $4$-tuple $ (X, \dist, \treshold_{\ell}, \treshold_{u})$, where $X$ is the set of points. 
The function $\dist \colon X \times X \rightarrow \mathbb{R}_{\geq 0}$ is a symmetric function representing distances between the points.
Additionally, we have symmetric functions 
$\treshold_{\ell} \colon X \times X  \to  \mathbb{R}_{\geq 0}$ 
and 
$\treshold_{u} \colon X \times X  \to  \mathbb{R}_{\geq 0}$ 
over the pairs of points. 
Besides symmetry,
we assume that
$\dist(x,x)=\treshold_\ell(x,x)=\treshold_u(x,x)=0$ for all $x\in X$.
Then the  \Cuvd problem is to decide whether a given $4$-tuple $ (X, \dist, \treshold_{\ell}, \treshold_{u})$ could be transformed into an ultrametric  $ \Xcal' = (X, \dist')$ by changing at most $k$ values of $\dist$ such that for every   pair $x,y\in X$ we satisfy $\treshold_{\ell}(x,y) \leq \dist'(x,y) \leq \treshold_{u}(x,y)$. Because of the symmetry of distances, we always assume that whenever $\dist(x,y)$ gets modified, the same operation adjusts $\dist(y,x)$.

An equivalent formulation of the problem uses distance matrices. Let \( \Xcal = (X, \dist) \) be a distance space where \( X = \{x_1, \ldots, x_n\} \). Then \( \Xcal \) is equivalently represented by the \emph{distance matrix}
$D = \{d_{ij}\}_{1\leq i,j\leq n},$ 
where each entry \( d_{ij} = \dist(x_i, x_j) \) for all \( i, j \in \{1, \ldots, n\} \).  
To measure closeness, we use the $\ell_0$ distance. Given two matrices $D$ and $D'$, the $\ell_0$ distance between them, denoted $\|D-D'\|_0$, is the number of non-zero entries of $D-D'$.
In other words, a distance space $\Xcal = (X, \dist)$ with distance matrix $D$ can be transformed into an ultrametric $\Xcal' = (X, \dist')$ with distance matrix $D'$ by modifying at most $k$ distances if and only if $\|D - D'\|_0 \leq 2k$. The multiplicative factor $2$ arises because the distance matrix $D$ is symmetric, meaning that changing one distance in $\Xcal$ affects two entries in $D$.

\defparproblema{\cuvd (\Cuvd)}{$4$-tuple $ (X, \dist, \treshold_{\ell}, \treshold_{u})$ with the corresponding  distance matrix $D$, integer $k$}%
{$k$}%
{
Decide whether there exists an ultrametric space $ \Xcal' = (X, \dist')$ on same set of points $X$ with the distance matrix $D'$
such that $\|D - D'\|_0 \leq 2k$ and such that for every pair $x,y\in X$, we have $\treshold_{\ell}(x,y) \leq \dist'(x,y) \leq \treshold_{u}(x,y)$.
}

\Cuvd could be seen as a ``violation variant'' of the  \textsc{Matrix Sandwich} problem   \cite{farach1995robust}.

\medskip
 Our first result is the theorem establishing a polynomial kernel for \Cuvd.

 \begin{restatable}{theorem}{KernelConstrUltr}
\label{theorem:kernelconstrultr}
\Cuvd{} admits a kernel with $ 3k^2 + 2k $ points and of size $\cO(k^4\log{k})$.
\end{restatable}
\noindent
Here, the size refers to the total encoding length of the instance.
Since the reduced instance has $\Ocal(k^2)$ points, it gives rise to $\Ocal(k^4)$ pairwise distances and thresholds, and each of which may require $\Ocal(\log k)$ bits to encode, resulting in a total size of $\Ocal(k^4 \log k)$ bits.

In other words, \Cref{theorem:kernelconstrultr} establishes a polynomial-time algorithm that, for a given instance of the problem, constructs an equivalent instance on $\mathcal{O}(k^2)$ points. 
The algorithm is based on a sequence of reduction rules, which is a common approach in kernelization 
(see, e.g.,~\cite{fomin2019kernelization}). However, the reduction rules here differ substantially 
from the classical techniques (such as those based on the Sunflower Lemma) that are typically used 
for $d$-\textsc{Hitting Set}. 
The key idea behind the rules is a carefully designed preprocessing of distances and edge constraints 
within ``bad'' triangles, which are roughly triples of points that violate the ultrametric properties. 
This preprocessing allows us to identify ``safe'' points whose removal yields an equivalent, 
but smaller, instance. 

 \Cref{theorem:kernelconstrultr} implies an \classFPT algorithm with running time 
\(k^{\mathcal{O}(k)} \cdot |X|^{\mathcal{O}(1)}\) for solving \Cuvd where $|X|$ is the number of points. 
The idea is to brute-force over all possible sets of at most $k$ pairs, resulting in 
$\binom{\mathcal{O}(k^4)}{k}$ possible guesses. 
For each such guess, one can determine in polynomial time (this requires some additional arguments, see \Cref{lemma:poly-if-known-edges}) 
whether modifying the distances between the chosen pairs transforms the instance into an ultrametric. 
In the next theorem, we present a single-exponential \classFPT algorithm.

\begin{restatable}{theorem}{FPTCuvd}
    \label{theorem:fpt-cuvd}
    \Cuvd{} admits an \classFPT{} algorithm with running time $\Ocal(9^k \cdot k \cdot |X|^{2}\big)$.
\end{restatable}

The running time provided in \Cref{theorem:fpt-cuvd} is asymptotically tight in the following sense. 
First, the exponential dependence on $k$ cannot be improved to subexponential. Indeed, as we already noted, even \textsc{Cluster Editing}, which is a special case of \Cuvd, does not admit a subexponential-time algorithm of the form $2^{o(k)} \cdot n^{\mathcal{O}(1)}$ unless the Exponential Time Hypothesis (ETH) fails~\cite{FominKPPV14,komusiewicz2012cluster}. 
Second, the size of the input of the problem is $\Ocal(|X|^2)$. 
Thus, for constant $k$, the running time of the theorem is \emph{linear} in the input size.
 
The algorithm in \Cref{theorem:fpt-cuvd} is based on branching and follows what is known in the 
literature on Hierarchical Clustering as the \emph{divisive} paradigm. 
At each branching step, the algorithm solves a suitable variant of the \textsc{Cluster Editing} 
problem by branching on all induced copies of a path $P_3$ (the forbidden structure for cluster graphs). 
However, since the algorithm proceeds through multiple levels of branching, the values of certain edges 
may be altered multiple times during the computation. 
This differs from the standard branching scheme for \textsc{Cluster Editing}, and it complicates both 
the branching rules and the analysis of the running time. 
In particular, the analysis requires a parameter different from $k$ to  bound the running time 
of the algorithm.

\ifshort

\else

\subsection{Overview of our methods}
As we already mentioned, \Cuvd generalizes \textsc{Cluster Editing}. 
A fixed-parameter tractable (\classFPT) algorithm (with parameter $k$) for \textsc{Cluster Editing} can be obtained by a simple hitting set branching on all $P_3$'s~\cite{cygan2015parameterized}.  
It might seem that for \uvd, a similar simple hitting set problem could be formulated for a family of violated triples, specifically, the triples of distances that do not satisfy the ultrametric inequality. This would suggest the possibility of an easy branching algorithm to establish that the problem is in \classFPT.
However, even if one identifies a set of edges that hits all the violated triples, there may be no way to modify the distances on these edges without creating new violations, as shown in \Cref{fig:ultrametric-eg}.
This distinction significantly increases the difficulty of \uvd compared to a hitting set problem on bounded-size cycles.

\begin{figure}[h]
    \centering
    \includegraphics[width=0.4\textwidth]{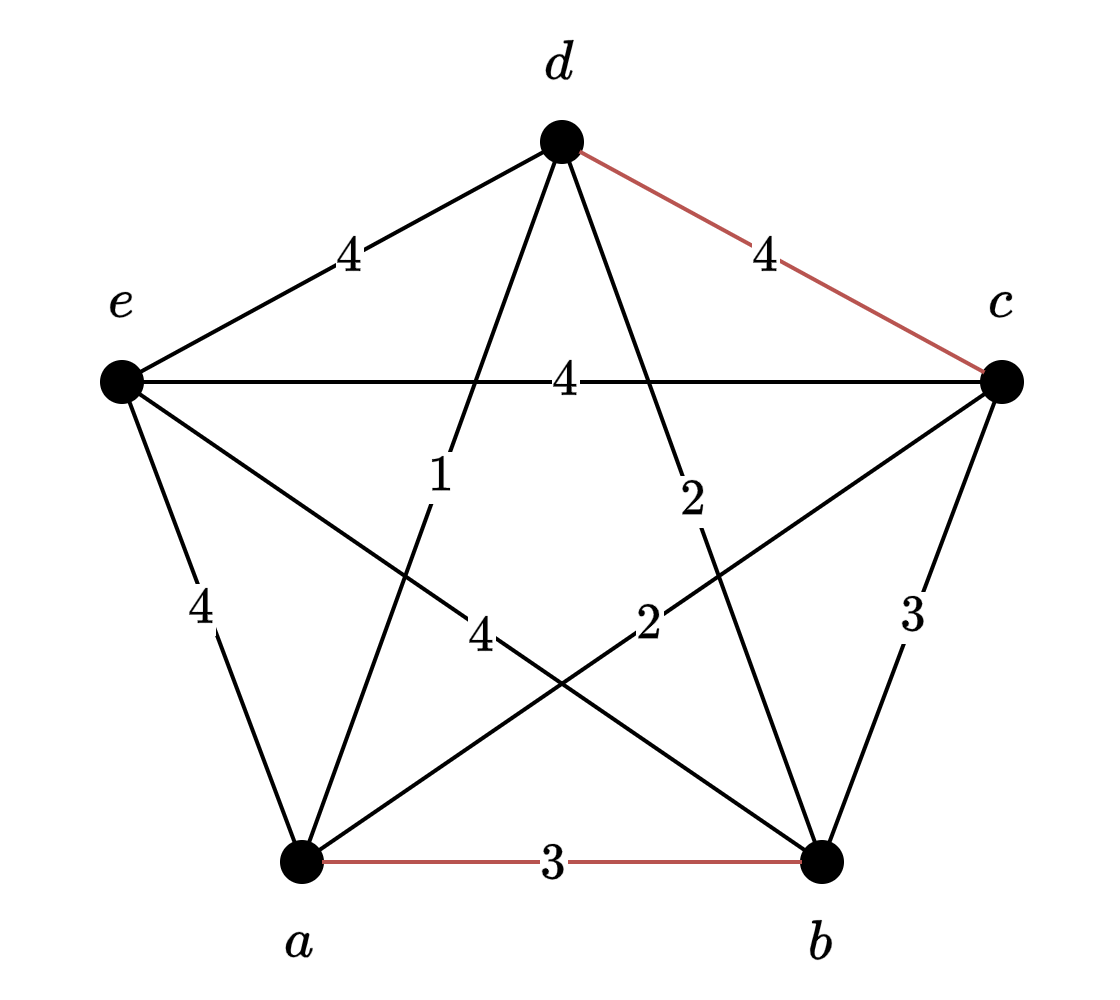} 
    \caption{In the figure $abd, acd, bcd$ are the violated triples. 
    However, eliminating these triples by assigning new distances to $ab$ and $cd$  creates new violated triples.}
    \label{fig:ultrametric-eg}
\end{figure}

This difference from the classic \textsc{$d$-Hitting Set} problem means that the
powerful toolbox developed for fast parameterized algorithms and
kernelization, see e.g. \cite{cygan2015parameterized,fomin2019kernelization},  cannot be directly applied to \Cuvd. 

    \ifshort
    
    We next establish a key structural property of solutions that will be used repeatedly in both the kernelization and the FPT algorithm. 
    Intuitively, if we guess the set of edges whose distances are modified in an optimal solution, then the remaining task of assigning consistent distances can be carried out efficiently while restricting attention to solutions whose values come from a finite set determined by the input. 
    The following proposition formalizes this idea.

    \begin{proposition}
        \label{lemma:poly-if-known-edges}
        There is an algorithm which when given an instance of \Cuvd $(\Xcal=(X, \dist, \treshold_{\ell}, \treshold_{u}), k)$ with corresponding distance matrix $D$ and set of edges $|A| \leq k$ will decide whether one can convert $D$ into an ultrametric $D'$ by changing only the edges in $A$ in polynomial time such that $D'$ only uses values from $D \cup \Treshold_{u} \cup \Treshold_{\ell}$. 
    \end{proposition}

    \begin{proof}
        The  idea is to set the value of $\dist$ to every edge in $A$ equal to its upper-threshold value  and then monotonically decrease some of these values until the ultrametric condition is satisfied. By modifying thresholds,  we preserve the invariant that each edge in $A$ always attains its current upper threshold.  Hence, at any step, either the current distances already form an ultrametric, or we could further  decrease certain edge values.
    
        We start by setting for every edge $xy \in A$,  $\dist(x,y) := \treshold_{u}(x,y)$. After this, we follow the following greedy procedure.
        If $D$ becomes an ultrametric, we stop and return $D$ as the solution. If $D$ is not an ultrametric, then there exists a triangle $T = xyz$ which does not satisfy the ultrametric conditions. 
        By symmetry, we can assume that $\dist(x,y) > \dist(y,z) \geq \dist(x,z)$.
        The weights of the edges $yz$ and $xz$ cannot be increased: if these edges are not in $A$, then we do not change them, and if they are in $A$, their weights are already at their upper thresholds and therefore cannot be increased further.
        Thus, we have to decrease $\dist(x,y)$. If $xy\notin A$ then we cannot modify $\dist(x,y)$ and conclude that $D$ cannot be converted into ultrametric. 
        Otherwise, if $xy\in A$, and in every solution $\dist(x,y)$ should be at most $\dist(y,z)$. Hence, it is safe to decrease $\dist(x,y)$ to $\dist(y,z)$ and set a new threshold $\treshold_{u}(x,y) := \dist(y,z)$.    

        Firstly, in the above procedure only entries from $D \cup \Treshold_u \cup \Treshold_l$ are used.
        Additionally, the value of any edge $xy \in A$ can be changed at most $|D \cup \Treshold_u \cup \Treshold_l| = 3\cdot|X|^{2}$ times.       
        Also, we can check whether a given distance matrix is ultrametric in time $\mathcal{O}(|X|^2)$ \cite{KANNAN199626}.
        Thus, the above procedure takes polynomial time.
    \end{proof}

    \Cref{lemma:poly-if-known-edges} implies the following property of an optimal solution, which we will use in obtaining the kernel as well as in the algorithm.

    \begin{corollary}
        \label{corollary:tight-solution}
        If the given instance $(\Xcal=(X, \dist, \treshold_{\ell}, \treshold_{u}), k)$ with corresponding distance matrix $D$ of \Cuvd is a \yesinstance then there exists a solution ultrametric space $ \Xcal' = (X, \dist')$ on the same set of points $X$ with distance matrix $D'$ such that $\|D - D'\|_0 \leq 2k$, satisfying the threshold conditions, and
        each entry of $D'$ is from the set $D \cup \Treshold_{u} \cup \Treshold_{\ell}$.    
    \end{corollary}
    
    \begin{proof}
        Since $\Xcal$ is a \yesinstance, there is a set $A$ of at most $k$ edges whose weight changes transform $\Xcal$ into an ultrametric.  We apply  \Cref{lemma:poly-if-known-edges} to this set $A$.
    \end{proof} 

    The corollary guarantees the existence of an optimal solution whose distances come from a finite set of candidate values derived from the input. 
    This discretization is key for both the kernelization and the algorithm, as it allows us to avoid arbitrary real values and instead work with a bounded search space.

    \else

    \fi

\medskip\noindent\textbf{Proof of \Cref{theorem:kernelconstrultr}}
We start by outlining the proof strategy behind the kernelization algorithm in  \Cref{theorem:kernelconstrultr} and highlighting the main intuitions.

For a given instance $\bigl((X,\dist,\tau_\ell,\tau_u),k\bigr)$ of \Cuvd, 
we often use
graph-theoretic terminology, referring to points of $X$ as \emph{vertices} and
pairs of points as \emph{edges}.

The proof proceeds 
  in three  phases:
\begin{enumerate}
  \item \emph{Normalize and bound local obstructions} by a sequence of reduction rules that detect immediate \noinstances and enforce that each edge participates in only $O(k)$ ``critical'' triangle configurations.
  \item \emph{Extract an irrelevant vertex} and delete it; iterating this yields a bound $|X|=O(k^2)$.
  \item \emph{Compress the numeric encoding} (distances and thresholds)  to obtain a polynomial kernel.
\end{enumerate}

\medskip
\noindent\textsl{``Bad triangles'' as the right notion of obstruction.}
A triangle violating the ultrametric inequality 
\[
\dist(x,z)\le \max\{\dist(x,y),\dist(y,z)\}\qquad\text{for all }x,y,z\in X,
\]
(equivalently: every triangle is either equilateral or isosceles with the two larger sides equal)
is a local witness of non-ultrametricity, but for kernelization, we need a refined notion of a \emph{bad triangle}, which is not only violated but also  accounts for thresholds and forced edits. 

We partition edges into three classes:
\begin{itemize}
  \item $\FIXED$(\emph{fixed}): $\tau_\ell(xy)=\tau_u(xy)=\dist(x,y)$ (value is forced and cannot change),
  \item $\UNFIX$(\emph{unfit}): $\dist(x,y)\notin[\tau_\ell(xy),\tau_u(xy)]$ (must be changed in any solution),
  \item $\NORM$(\emph{normal}): $\dist(x,y)\in[\tau_\ell(xy),\tau_u(xy)]$ and $\tau_\ell(xy)\neq\tau_u(xy)$ (currently feasible but flexible).
\end{itemize}This classification separating \emph{forced} edits $\UNFIX$ from \emph{optional}
edits $\NORM$ is crucial for $\ell_0$-budget $k$ accounting. All of our rules
except \emph{one} do \emph{not} reduce $|X|$ directly. Instead, they
\emph{tighten thresholds} or \emph{force edges to become fixed}. Tightening
constraints may leave $k$ unchanged (since cost is incurred only when a
distance is actually modified), and we modify the instance
toward exposing an \emph{irrelevant} vertex. A vertex is irrelevant if it is not
incident to an edge from $\UNFIX$ and is not part of a bad triangle (a notion
defined below). As we prove, an irrelevant vertex can be safely removed, and
irrelevant-vertex removal is the only reduction rule that decreases the number
of points in the instance.

We define a triangle $T=xyz$ to be \emph{bad with respect to edge $xy$} if it falls into one of the following patterns:
\begin{enumerate}
  \item[(i)]  
  $xy,yz,xz\in \NORM\cup\FIXED$ and $xyz$ violates the ultrametric condition.
  \item[(ii)]  
  $xy\in \NORM\cup\FIXED\cup\UNFIX$ and both other edges are in $\UNFIX$.
  \item[(iii)] 
  $xy\in \UNFIX$, $yz,xz\in \NORM\cup\FIXED$, and either
  (iiiA) The other two sides are unequal, or
  (iiiB) Their values are equal, but the upper threshold on $xy$ exceeds these values.
\end{enumerate}

\smallskip
\noindent\emph{Intuition why we need bad triangles.} By a sequence of reduction rules and careful analysis, we obtain an equivalent instance in which the number of bad triangles that can ``hang'' on a single edge is $\cO(k)$.
If an edge participates in too many bad triangles, then either a reduction rule becomes applicable (for example, forcing the edge to be fixed or tightening its thresholds), or we can argue via the budget: each modification of an edge can resolve only a bounded number of bad triangles. 
Hence, if an edge were contained in more than $\cO(k)$ bad triangles, then resolving all of them would require more than $k$ modifications, contradicting the assumption that the instance is a yes-instance. 
Therefore, after exhaustive application of the reduction rules, every edge can be involved in only $\cO(k)$ bad triangles.

\ifshort
    \medskip
    \noindent\textsl{Creating and deleting an irrelevant vertex.}
    We define a vertex $u$ to be \emph{irrelevant} if (i) it is not incident to any $\UNFIX$ edge, and
    (ii) for every edge $xy$, the triangle $uxy$ is \emph{not} bad with respect to $xy$.
    
    \smallskip
    \noindent\emph{Intuition why it works.}
    If a vertex $u$ is not contained in any bad triangle and has no incident edges from $\UNFIX$, then it does not participate in any local obstruction that would force modifications in a solution.
    In this situation, the distances from $u$ induce a well-structured layering of the remaining vertices.
    Let $d_1<\cdots<d_\ell$ be the distinct values among $\{\dist(u,x):x\neq u\}$ and define
    $
    X_i=\{x\in X\setminus\{u\} : \dist(u,x)=d_i\}.
    $
    The absence of bad triangles imposes strong constraints on this structure:
    \begin{itemize}
      \item there are no $\UNFIX$-edges between different layers (otherwise we would obtain a bad triangle of type (iiiA)),
      \item for $x\in X_i$ and $y\in X_j$ with $i<j$, the distance $\dist(x,y)$ must be equal to $d_j$ (otherwise $uxy$ becomes a bad triangle),
      \item within each layer $X_i$, distances can be modified independently, and any solution can be adjusted so that all distances are at most $d_i$ without increasing the number of edits.
    \end{itemize}
    Thus, the instance decomposes into independent subinstances on the layers, together with fixed distances between layers.
    As a consequence, any solution on $X\setminus\{u\}$ can be extended to a solution on $X$ without increasing the budget.
    Hence, the vertex $u$ does not influence the existence of a solution and can be safely removed.
    
    \smallskip
    \noindent This observation leads to the following reduction rule.
    
\else

\medskip
\noindent\textsl{Creating and deleting an irrelevant vertex.}
We define a vertex $u$ to be \emph{irrelevant} if (i) it is not incident to any $\UNFIX$ edge, and
(ii) for every edge $xy$, the triangle $uxy$ is \emph{not} bad with respect to $xy$.
We then apply the \emph{irrelevant-vertex removal} rule: delete such a $u$.

\smallskip
\noindent\emph{Intuition why it works.}
If $u$ is not in any bad triangle and has no  incident edges of $\UNFIX$, then the distances from $u$ induce a clean ``layering'':
letting $d_1<\cdots<d_\ell$ be the distinct values among $\{\dist(u,x):x\neq u\}$ and
$X_i=\{x:\dist(u,x)=d_i\}$, the absence of bad triangles forces:
\begin{itemize}
  \item no $\UNFIX$-edges between different layers (otherwise $uxy$ would be bad of type (iiiA)),
  \item for $x\in X_i$, $y\in X_j$ with $i<j$, necessarily $\dist(x,y)=d_j$ (otherwise $uxy$ becomes bad),
  \item within each $X_i$, any solution can be ``truncated'' so internal distances never exceed $d_i$ without increasing the number of edits.
\end{itemize}
Because of that,  any solution on $X\setminus\{u\}$ can be extended to a solution on $X$ by solving each layer independently and stitching layers together using the forced cross-layer distances. Hence deleting $u$ preserves yes/no equivalence.

\fi

\ifshort

\begin{leftbar}
    \begin{dtredrule}[Irrelevant Vertex Removal] \label{DTRR:10}
        Let $u$ be a vertex  such that no \unfix edge is incident with it and there is no edge $xy$ such that $T = uxy$ is a bad triangle \wrt $xy$.          Delete $u$.  
    \end{dtredrule}
\end{leftbar}

\begin{claim}\label{claim:DTRR:10} 
    \Cref{DTRR:10} is sound.
\end{claim}

\begin{proof}
    Suppose that in the instance $\Xcal=(X, \dist, \treshold_{\ell}, \treshold_{u})$ there exists a vertex $u\in X$ satisfying the conditions of the rule.
    
    Let $d_1 < d_2 < \cdots < d_\ell$ be the distinct values among $\{\dist(u,x) : x \in X\setminus\{u\}\}$, and partition $X\setminus\{u\}$ into $\ell$ sets $X_1, X_2, X_3, \dots, X_\ell$ based on their distances from $u$. Formally,
    \[
    X_i = \{x \in X \mid \dist(u,x) = d_i\}.
    \]
    
    Since $u$ is not contained in any bad triangle with respect to any edge $xy$, the following properties hold:
    \begin{enumerate}[label=(\roman*)]
        \item no edge between $X_i$ and $X_j$ for $i \neq j$ belongs to $\UNFIX$, as otherwise $uxy$ would be bad of type (iiiA); 
        \item for $x \in X_i$ and $y \in X_j$ with $i<j$, we must have $\dist(x,y)=d_j$, since otherwise $uxy$ would be bad;
        \item for $x,y \in X_i$, either $\dist(x,y)\le d_i$, or $xy\in\UNFIX$ with $\tau_u(x,y)\le d_i$, as otherwise $uxy$ would be bad. 
    \end{enumerate}
    
    We prove that $(\Xcal,k)$ is a \yesinstance if and only if $(\Xcal'=(X\setminus\{u\},\dist,\treshold_\ell,\treshold_u),k)$ is a \yesinstance. The forward direction is immediate.
    
    For the reverse direction, suppose $(\Xcal',k)$ is a \yesinstance. Then there exist integers $k_1,\dots,k_\ell$ with $\sum_{i = 1}^{\ell} k_i \le k$ such that each induced instance on $X_i$ can be transformed into an ultrametric $(X_i,\dist_i)$ using at most $k_i$ modifications.
    

    We may assume that each $(X_i,\dist_i)$ satisfies $\max_{x,y\in X_i} \dist_i(x,y) \le d_i$. 
    Indeed, if some edge $xy$ has $\dist_i(x,y) > d_i$, then $xy$ must belong to $\NORM$. In particular, its original value satisfied $\dist(x,y) \le d_i$, so reaching a value above $d_i$ means that $xy$ was already modified. 
    Decreasing $\dist_i(x,y)$ back to $d_i$ preserves feasibility and does not increase the number of modified edges, as $xy$ remains a modified edge. 
    Repeating this for all such edges yields the claim.
    
    We now construct a distance function $\dist'$ on $X$ as follows:
    \begin{itemize}
        \item $\dist'(u,x) = \dist(u,x)$ for all $x \in X$;
        \item $\dist'(x,y) = \dist(x,y)$ for $x \in X_i$, $y \in X_j$, $i \neq j$;
        \item $\dist'(x,y) = \dist_i(x,y)$ for $x,y \in X_i$.
    \end{itemize}
    Clearly, $\dist'$ differs from $\dist$ in at most $\sum_{i=1}^{\ell} k_i \le k$ entries.
    
    It remains to verify that $(X,\dist')$ is an ultrametric. We show this by showing every triangle satisfies ultrametric condition. Let $x,y,z \in X$:
    \begin{itemize}
        \item  $x,y,z\in X_i$ for some $i$.  Since  $(X_i, \dist_i)$ is ultrametric,  $x,y,z$ satisfy the ultrametric condition in $\dist'$ too.
        \item  $x\in X_i$, $y\in X_j$, $z\in X_k$, $i<j<k$. Since we did not change the distances between these vertices, we have $\dist'(x,y) = d_j$ and $\dist'(x,z)= \dist'(y,z) = d_k > d_j$ satisfying the ultrametric conditions. 
        \item  $x,y\in X_i$, $z\in X_j$, $i<j$ or $i>j$. If $i>j$ then $\dist'(x,z) = \dist'(y,z) = d_i$ and $\dist'(x,y) \leq d_i$, satisfying the ultrametric condition.   If $i<j$ then $\dist'(x,z) = \dist'(y,z) = d_j$ and $\dist'(x,y) \leq d_i < d_j$, satisfying the ultrametric condition. In either case, we do not violate the ultrametric condition.      
        \item  $x \in X_i$, $y\in X_j$, $i<j$, and $z=u$. We have $\dist'(x,u) = d_i$ and $\dist'(y,u) = d_j$ Since we did not modify $xy$, we have that $\dist'(x,y) = d_j$ and $d_i < d_j$, this triangle satisfies the ultrametric condition.
        \item  $x,y \in X_i$, $z=u$. Since $\dist'_i \leq d_i$, and $\dist'(x,u) = \dist'(y,u) = d_i$, the triangle $xyu$ satisfy the ultrametric condition in $\dist'$.
    \end{itemize}
    
    Thus, $(X,\dist')$ is an ultrametric obtained using at most $k$ modifications, completing the proof.   
\end{proof}

\else

\fi

\medskip
\noindent\textsl{Bounding the number of vertices.}
Once irrelevant vertices are exhaustively removed, {every remaining vertex must participate in some bad triangle.}
Let $A$ be the (unknown) set of at most $k$ edges whose distances are modified by an optimal solution.
By  reduction rules,  each edge in $A$ can be charged with only $O(k)$ bad triangles, and each bad triangle contributes only a constant number of vertices beyond its charged edge.
A counting argument yields
$
|X|\le 3k^2+2k,
$. 
This gives a  reduced instance with $O(k^2)$ points.

\medskip
\noindent\textsl{Order-preserving scaling.}
Even after obtaining an equivalent instance with $|X| = \cO(k^2)$, the distances
and thresholds may still be numerically large. By exploiting properties of
ultrametrics and the bound $|X| = \cO(k^2)$ on the number of points in the reduced
instance, it is possible to use order-preserving scaling to construct an equivalent instance in which all
distances and thresholds are bounded by $\cO(k^4)$. This yields a  kernel of size $\cO(k^4 \log k)$.

\medskip

\medskip
\noindent\textbf{Proof of \Cref{theorem:fpt-cuvd}.}
We now sketch the ideas behind a single-exponential FPT algorithm for \Cuvd{} running in
$\Ocal(9^k\cdot k\cdot |X|^2)$ (\Cref{theorem:fpt-cuvd}).
The algorithm follows a \emph{divisive, top--down} paradigm that recovers an ultrametric by progressively enforcing the nested cluster structure that ultrametrics induce.

\medskip
\noindent\textsl{Ultrametrics as a hierarchy of cluster graphs.}
Fix a distance value $\lambda$. In an ultrametric $(X,\dist')$, the graph on $X$ with edges
$\{xy:\dist'(x,y)<\lambda\}$ is always a \emph{cluster graph} (a disjoint union of cliques).
Intuitively, points at distance $<\lambda$ must form ``clusters'' in which all pairs are close; otherwise an induced $P_3$ would witness a triangle where two sides are $<\lambda$ but the third is $\ge \lambda$, contradicting ultrametricity.
This observation suggests a strategy: process distance thresholds from large to small, and at each threshold enforce that the corresponding proximity graph becomes a cluster graph.

It is always possible (\Cref{corollary:tight-solution}) to restrict attention to solutions whose distances take values only from the finite set  $d_1<\cdots<d_L$ (with $L=\Ocal(|X|^2)$). We process levels from $L$ down to $1$.

At each level $i$, we consider the auxiliary graph $G_i(\dist,M)$ whose edges represent pairs that are currently
 {strictly below} $d_i$, together with a set $M$ of  {marked edges}.
In a target solution $\dist'$, marked edges are precisely those whose values we have committed to  {change} but whose final level is not yet decided.
If the instance is a \yesinstance{}, then for each level $i$ the graph $G_i(\dist',\emptyset)$ must be a cluster graph.
Hence, at level $i$ we can branch exactly as in \textsc{Cluster Editing}: whenever $G_i$ contains an induced $P_3$, at least one of its three pairs must be edited (turned into a non-edge or an edge) in the proximity graph.

\medskip
\noindent\emph{Main challenge: decreasing an edge may happen multiple times.}
In standard \textsc{Cluster Editing}, an edit is final: inserting or deleting an edge resolves the local obstruction permanently.
Here, an edit in $G_i$ corresponds to changing a \emph{distance}, and crucially:
\begin{itemize}
  \item \emph{Increasing} $\dist(x,y)$ to $d_i$ is final at level $i$ (we separate the pair at this scale).
  \item \emph{Decreasing} $\dist(x,y)$ below $d_i$ is \emph{not} final: we do not yet know whether the correct target is $d_{i-1}$, or $d_{i-2}$, etc.
\end{itemize}
This is the key complication compared to ordinary \textsc{Cluster Editing} and is the reason the algorithm introduces  {marked edges}.
Marking is an intermediate commitment: ``this pair must get smaller than $d_i$, but we postpone deciding how small until we reach lower levels.''

\medskip
\noindent\textsl{Marked edges and violated triangles.}
At level $i$, we partition edges (pairs of vertices)  into:
\[
D_i=\{xy:\dist(x,y)=d_i\},\qquad D_{<i}=\{xy:\dist(x,y)<d_i\},\qquad M=\text{marked}.
\]
A \emph{violated triangle at level $i$} is a triple $xyz$ with  {exactly one} edge in $D_i$ and the other two in $D_{<i}\cup M$.
Such a triple is precisely an induced $P_3$ in the auxiliary graph $G_i$: two pairs are ``close'' at scale $d_i$ (edges of $G_i$), but the third is not.
Fixing violated triangles is how we enforce the cluster-graph property level-by-level.

\medskip
\noindent\textsl{Structural lemma enabling the divisive recursion.}
A crucial invariant is:
\begin{quote}
If there are no violated triangles at levels $j\in\{i,\ldots,L\}$, then $G_i$ is a cluster graph.
Moreover, any violated triangle at a lower level $\ell<i$ must lie entirely within a single clique of $G_i$.
\end{quote}
This lemma formalizes the divisive paradigm: once $G_i$ is a cluster graph, its cliques define a partition of $X$, and all remaining ``finer-scale'' inconsistencies happen {inside} these cliques.
Thus we can safely descend from level $i$ to $i-1$ and continue the search without needing to reconsider interactions between different cliques.

\medskip
\noindent\textsl{Branching and measure.}
Given a violated triangle $xyz$ at level $i$ (equivalently, a $P_3$ in $G_i$), at least one of the three pairs must change its status at this level.
The algorithm branches into   cases, choosing which pair to ``edit'':
All branches additionally respect thresholds.

A second core challenge is to obtain a clean single-exponential bound despite the fact that a distance can be decreased multiple times across levels.
To control this, we use a \emph{measure} $\mu$ that charges:
\begin{itemize}
  \item $2$ units for a final increase ($D_{<i}\to D_i$),
  \item $1$ unit for marking ($D_i\to M$),
  \item $1$ unit for finalizing a marked edge ($M\to D_i$).
\end{itemize}
Thus, changing the value of one pair from its original distance to its final chosen distance costs at most $2$ units of $\mu$.
Initializing $\mu=2k$ (and accounting for initially forced changes) ensures every branch decreases $\mu$ by at least $1$.
Hence the search tree has at most $3^\mu\le 3^{2k}=9^k$ leaves.

\fi

\subsection{Related Work}

\subparagraph{Cluster Editing.}
The special variant of \textsc{Ultrametric Violation Distance} when the values of the entries of the input distance matrix $D$ are in $\{1,2\}$ is the well-studied  
\textsc{Correlation Clustering}, also known as \textsc{Clustering with Qualitative Information},  or \textsc{Cluster Editing}. In this problem one has to cluster a set of objects based only on 
qualitative information concerning similarity between pairs of items. For every pair of objects, we have an indication if the objects are similar (+) or not (-). The task is to find a partition of the objects into clusters minimizing the amount of similarities between different clusters and non-similarities   inside of clusters.   The problem was introduced by Ben-Dor, Shamir,  and   Yakhini \cite{Ben-DorSY99} motivated by some problems from computational biology, and, independently, by Bansal,  Blum,  and Chawla \cite{Bansal04}, motivated by machine learning problems concerning document clustering according to similarities. 
The problem is known to be NP-complete \cite{ShamirST04}. Moreover, it was   shown by Charikar, Guruswami, and Wirth   \cite{CharikarGW03} that the problem is APX-hard.
It was intensively studied in approximation algorithms community, see e.g. \cite{AilonCN08,AlonMMN05,AroraBKSH05,CharikarGW05j,CharikarW04,GiotisG06t,ShamirST04}. 
Recent advances on this problem include \cite{cohen-addad2022correlation,cohen-addad2024combinatorial,cao2024understanding,10353161,10.1145/3717823.3718181}.

The graph-theoretic formulation of \textsc{Correlation Clustering} is known as \textsc{Cluster Editing}, which is the name commonly used in the parameterized algorithms community.  
The parameterized version of \textsc{Cluster Editing}, along with its variants, has been studied extensively~\cite{BockerD11,BodlaenderFHMPR10,Damaschke10,FellowsGKNU11,FominKPPV14,GrammGHN05,GuoKKU11,GuoKNU10,komusiewicz:sofsem,ProttiSS09}.  
The problem is solvable in time $\mathcal{O}(1.62^k + n + m)$~\cite{bocker2012golden}, and it admits a kernel with $2k$ vertices~\cite{CaoC10,ChenM10}.  

Moreover, the problem does not admit a subexponential-time algorithm of the form $2^{o(k)} \cdot |X|^{\mathcal{O}(1)}$ unless the Exponential Time Hypothesis (ETH) fails~\cite{FominKPPV14,komusiewicz2012cluster}. Since \textsc{Cluster Editing} is a special case of the \textsc{Ultrametric Violation Distance} problem, this lower bound naturally carries over to the latter as well.

\subparagraph{Fitting Distance by Ultrametrics.}
In the numerical taxonomy problem, we are given a matrix $D$ of measured pairwise distances $\dist$
for a set of elements $X$, and the goal is to produce a tree metric or an ultrametric with matrix $D'$  with distance function $\dist'$ that spans \( X \) and minimizes the \( \ell_p \)-norm
\[
\|D - D'\|_p := \left( \sum_{\{i,j\} \in \binom{X}{2}} |\dist(i,j) - \dist'(i,j)|^p \right)^{1/p},
\]
where \( p \) is a   constant with \( 1 \leq p \leq \infty \).

Since its introduction in 1967 by Cavalli-Sforza and Edwards~\cite{cavalli1967phylogenetic},   the numerical taxonomy problem has collected an extensive literature. While this problem was initially introduced in the \( \ell_2 \)-norm, Farris suggested using the \( \ell_1 \)-norm in 1972 \cite{farris1972estimating}.  
For the \( \ell_\infty \)-norm, it is known that an optimal ultrametric can be computed in time proportional to the number of input distance pairs \cite{farach1993robust}. Approximation algorithms for fitting tree metrics with the \( \ell_\infty \)-norm are studied in \cite{agarwala1999approximability}.
For constant \( p \) with \( 1 \leq p < \infty \), the developments have been slower and remain much less understood to date. 
Among them, the \( \ell_1 \)-norm in particular has been extensively studied, and a constant-factor approximation was given~\cite{10.1145/3639453}.

For \( 1 < p < \infty \), an \( O(\log n \log \log n)^{\frac{1}{p}} \) approximation ratio remains the best known \cite{ailon2011fitting}.

In parameterized complexity, the only norm studied prior to our work was the $\ell_1$-norm 
(the problem was referred to there as \textsc{Hierarchical Clustering}). 
Hartung, Guo, Komusiewicz, Niedermeier, and Uhlmann~\cite{HartungGKNU10}, 
under the assumption that the entries of $D$ are positive integers, proved that the problem 
admits a kernel with $\mathcal{O}(k^2)$ points, where 
$
  k = \|D - D'\|_1.$
They also gave a kernel with $2k \cdot (M + 2)$ elements, where $M$ is the maximum distance in $\dist$. 
Later, Cao and Chen~\cite{CaoChen2013hierarchical} improved both the running time of the algorithm 
and the size of the kernel.

As noted by Charikar and Gao~\cite{charikar2024improved}, given the long history of the study 
of this family of ultrametric fitting problems, it is somewhat surprising that the natural problem 
of minimizing the $\ell_0$-norm of the error (under the name \textsc{Ultrametric Violation Distance}) 
was introduced and studied only very recently by Cohen-Addad, Fan, Lee, and Mesmay~\cite{CohenAddadFLM25} 
(the conference version of their result appeared in~\cite{cohen2022fitting}). 
The work of Cohen-Addad et al.\ sparked significant interest in approximation algorithms 
for \textsc{Ultrametric Violation Distance}~\cite{bathie2025,charikar2024improved,carmel2025fitting,11369173,DasKipouridisSpoerhase2026}. 

More generally, \textsc{Ultrametric Violation Distance} lies at the intersection of two well-studied classes of problems: 
\textsc{Hierarchical Clustering}~\cite{ward1963hierarchical,murtagh2012algorithms} and 
\textsc{Metric Violation Distance}~\cite{FanRB18,gilbert2017sparse,bentert2025distances,fan_et_al:LIPIcs.SWAT.2020.25},
where one is allowed to change at most $k$ distances in order to transform a distance space into a metric.

Fan, Gilbert, Raichel, Sonthalia, and Van~Buskirk~\cite{fan_et_al:LIPIcs.SWAT.2020.25} show that \textsc{Metric Violation Distance} can be solved in time $k^{\mathcal{O}(k)} n^{\mathcal{O}(1)}$, where $k$ is the number of required modifications. We note that their algorithm applies to a more general setting. 
Since \uvd can be reduced to an instance of \textsc{Metric Violation Distance} by tropicalizing the distances via the map $d \mapsto n^d$; see, for example, \cite[Corollary~5.16]{CohenAddadFLM25}, the algorithm of Fan et al.\ can be used to solve \uvd in time $k^{\mathcal{O}(k)} n^{\mathcal{O}(1)}$.\footnote{We are grateful to a referee for pointing out this connection.} 
The result of Fan, Gilbert, Raichel, Sonthalia, and Van~Buskirk is not directly comparable to our work. While they establish fixed-parameter tractability for a more general problem, none of our algorithmic contributions, specifically, the polynomial kernelization and the single-exponential-time algorithm, could be derived from their work.

\ifshort

\else

\section{Preliminaries and basic facts about ultrametrics}\label{sec:plelim}
    We use 
    \( \mathbb{R}_{\geq 0} \) 
    to denote the set of non-negative 
    real numbers. 

    Let $X$ be a set. A function $\dist \colon X\times X \to \mathbb{R}_{\ge 0}$  is a \emph{distance} on $X$ if: 
    \begin{itemize}
        \item[(i)] $\dist$ is symmetric, that is,  for any $x, y \in X$, $\dist(x, y) = \dist(y, x)$, and  
        \item[(ii)] $\dist(x, x) = 0$ for all  $x\in X$.
    \end{itemize}
    Then, $ (X, \dist)$ is called a \emph{distance space}.

    If, in addition function $\dist$ satisfies a triangle inequality: $\dist(x, z) \le \dist(x, y) + \dist(y, z)$, for any $x, y, z \in X$, then $\dist$ is called a pseudometric on $X$. 
    Also, if $\dist(x,y)=0$ only for $x=y$, then $\dist$ is called a \emph{metric}, and $ (X,\dist)$ a \emph{metric space}. In this paper we are interested in special metrics.

    \begin{definition}[Ultrametric]
        A metric space $ (X,\dist)$ is an  \emph{ultrametric space} if for every distinct $x,y,z$, 
        \[
        \dist(x,z) \leq \max\{\dist(x,y),\dist(y,z)\}.
        \]
        We refer to the elements of $X$ as points.
    \end{definition}

An alternative way to view a distance space is as a complete graph $G[X]$ on the vertex set $X$, where the distances correspond to the weights of the edges of $G$. Throughout the paper, we adopt graph-theoretic terminology, referring to the points of $X$ as \emph{vertices} and pairs of points as \emph{edges}. Moreover, we often refer to the distance $\dist(x,y)$ between points $x,y \in X$ as the \emph{weight} of the edge $xy$. We use the same convention for the threshold functions $\treshold_u$ and $\treshold_\ell$, and see them as functions on the edges.  
In particular, throughout the paper, we assume that whenever the distance $\dist(x,y)$ ($\treshold_\ell(x,y)$ or $\treshold_u(x,y)$, respectively) is modified in our algorithms, the same modification concerns $\dist(y,x)$ ($\treshold_\ell(y,x)$ or $\treshold_u(y,x)$).
Finally, we refer to a cycle of length $3$ in the graph (equivalently defined by three vertices or three edges) as a \emph{triangle}.

    \begin{definition}[Values $\Dvalues$, $\Treshold_{u}$, and  $\Treshold_{\ell}$]\label{def:thrvalues}
        Let  $(\Xcal=(X, \dist, \treshold_{\ell}, \treshold_{u}), k)$ be an instance of \Cuvd and $D$ be the corresponding distance matrix. We use $\Treshold_{u}$ (and  $\Treshold_{\ell}$)  to denote the set of upper (and lower) thresholds for all edges (pair of points) in $X$. In other words,  
        \[\Treshold_{u} = \{\treshold_{u}(e) \mid e\in X \times X\},\] 
        and \[\Treshold_{\ell} = \{\treshold_{\ell}(e) \mid e\in X \times X\}.\]
        We also use $\Dvalues$ to denote the set of all values of the distance matrix $D$.
    \end{definition}
    
    \medskip

At the end of this section, we summarize some simple facts about ultrametrics that will be used to obtain our main results.  
The first is \Cref{lemma:poly-if-known-edges}. In this setting, we are given a set of $k$ edges, and the question is whether it is possible to obtain an ultrametric by modifying only these selected edges. \Cref{lemma:poly-if-known-edges} provides a polynomial-time algorithm for this scenario. 

\Cref{lemma:poly-if-known-edges}  could be derived from the algorithm of Farach, Kannan, and Warnow~\cite{farach1995robust}; see \cite[Theorem~4]{farach1995robust}.  
Their algorithm applies to the more general \textsc{Ultrametric Graph Sandwich} problem. 
For self-containment, we provide a simpler proof of \Cref{lemma:poly-if-known-edges}
in the appendix.

    \begin{proposition}
        \label{lemma:poly-if-known-edges}
        There is an algorithm which when given an instance of \Cuvd $(\Xcal=(X, \dist, \treshold_{\ell}, \treshold_{u}), k)$ with corresponding distance matrix $D$ and set of edges $|A| \leq k$ will decide whether one can convert $D$ into an ultrametric $D'$ by changing only the edges in $A$ in polynomial time. Moreover, the values in the solution matrix $D'$ could be selected from the set $\Dvalues \cup \Treshold_{u} \cup \Treshold_{\ell}$.
    \end{proposition}

    \begin{proof}
        Refer to Appendix~\ref{app:poly-if-known-edges} for the proof.
    \end{proof}

    \Cref{lemma:poly-if-known-edges} implies the following property of an optimal solution, which we will use in obtaining the kernel as well as in the algorithm.

    \begin{corollary}
        \label{corollary:tight-solution}
        If the given instance $(\Xcal=(X, \dist, \treshold_{\ell}, \treshold_{u}), k)$ with corresponding distance matrix $D$ of \Cuvd is a \yesinstance then there exists a solution ultrametric space $ \Xcal' = (X, \dist')$ on the same set of points $X$ with distance matrix $D'$ such that $\|D - D'\|_0 \leq 2k$, satisfying the threshold conditions, and
        each entry of $D'$ is from the set $\Dvalues \cup \Treshold_{u} \cup \Treshold_{\ell}$.    
    \end{corollary}
    
    \begin{proof}
        Since $\Xcal$ is a \yesinstance, there is a set $A$ of at most $k$ edges whose weight changes transform $\Xcal$ into an ultrametric.  We apply  \Cref{lemma:poly-if-known-edges} to this set $A$.
    \end{proof}

    We now prove a ``truncation'' property of ultrametrics that will be used in obtaining the kernel.

    \begin{lemma}
        \label{lemma:truncation}
        Let $(X,\dist)$ be an ultrametric and $\lambda \geq 0$. Then for 
     $\dist'(x,y) = \min\{\dist(x,y),\lambda\}$ the distance space $(X,\dist')$  is also an ultrametric.
    \end{lemma}

    \begin{proof}
    For contradiction, assume that $(X,\dist')$ contains a violated triangle $T=xyz$. Since $(X,\dist)$ is ultrametric, at least one edge of $T$ has to be truncated to $\lambda$.
    Thus $T$ contains at least one edge of weight $\lambda$. 
    However, it cannot be two or three edges of $T$ that have weight $\lambda$, because then $T$ is not violated. Hence exactly one edge of $T$, say $xy$, has weight $\lambda$, while the two others have weight smaller than $\lambda$. But then in $(X,\dist)$ we have 
\[
\dist(x,y)\geq \lambda > \max\{\dist(x,z), \dist (y,z)\},
\]
violating the assumption that $(X,\dist)$ is an ultrametric.
    \end{proof}

\section{Polynomial kernel for \Cuvd}\label{sec:kernel}

In this section we prove \Cref{theorem:kernelconstrultr}, which provides a polynomial kernel for \Cuvd.
Before we proceed with the proof, we need to set up several definitions and notions that will be used in it.
As with many kernelization algorithms, the proof of the theorem proceeds via a sequence of reduction rules.
Our main reduction rule identifies an ``irrelevant'' vertex (point) of $X$ whose removal yields an equivalent instance.
However, a priori the input, while large, may not contain such a vertex.
Here our reduction deviates from the common approach to kernelization.
Most of our rules do not change the size of the instance; rather, they modify edge thresholds towards  creating an irrelevant vertex at the end.
The procedure for modifying thresholds is nontrivial and requires careful analysis.
 
Depending on how an edge distance meets the threshold conditions,  
we partition the edge set into three classes defined below.
    \begin{itemize}
        \item[$\FIXED$:] \emph{\Fixed edges}. These are the edges whose distances are fixed and cannot be changed. Specifically, for an edge $xy$, $ \dist(x,y) = \treshold_{\ell}(x,y) = \treshold_{u}(x,y)$.
        \item[ $\UNFIX$:] \emph{\Unfix edges}. These are the edges whose distances have to be changed because their current distance does not lie within the threshold. Formally, either $\dist(x,y) < \treshold_{\ell}(x,y)$ or $\dist(x,y) > \treshold_{u}(x,y)$.   
        \item[$\NORM$:] \emph{\Normal edges.} These are the remaining edges,  that is, the edges that are not \Fixed and \Unfix. Formally, $\treshold_{\ell}(x,y) \leq \dist(x,y) \leq \treshold_{u}(x,y)$ and $\treshold_{\ell}(x,y) \neq \treshold_{u}(x,y)$.
    \end{itemize}

Each reduction rule either (i) identifies the instance as a \noinstance,  (ii) changes an edge threshold,  (iii) changes the status of an edge, and 
(iv) deletes a point.  
A \normal\ edge can become either \fixed\ or \unfix.  
An \unfix\ edge may remain \unfix\ (possibly with a smaller upper threshold)  
or be promoted to \fixed.

Whenever an edge changes status, we must update the parameter~$k$ with care.  
This is the key distinction from the kernelization in the $\ell_1$-norm.  
In the $\ell_1$ setting, \emph{every} status change strictly decreases the parameter.  
By contrast, under the $\ell_0$-norm, multiple threshold adjustments may leave $k$ unchanged.
Accordingly, in our kernelization algorithm we proceed as follows:
\begin{itemize}
  \item Turning a \normal\ edge into an \unfix\ edge, or decreasing the threshold of an existing \unfix\ edge, \emph{does not} reduce~$k$.
  \item Whenever a \normal\ or \unfix\ edge becomes \fixed, we decrease~$k$ if we change the distance. 
\end{itemize}

In order to define the reduction rules, we need the notion of a bad triangle. 
Recall that a \emph{triangle} is a triple of vertices, or equivalently of edges. Intuitively, a violated triangle $xyz$ is a triangle whose edges do not satisfy the ultrametric condition 
$
\dist(x,z) \leq \max\{\dist(x,y), \dist(y,z)\}.
$

However, for kernelization we need a more refined notion of a bad triangle. As we shall prove, a vertex that does not belong to any bad triangle and is not incident to any unfit edge is irrelevant and can be safely removed. Informally, all reduction rules but one will then guarantee that in a sufficiently large \yesinstance it is possible to reassign edge weights and thresholds to enforce the existence of an irrelevant vertex.

    \begin{definition}[Bad triangle \wrt edge $xy$] \label{def:badtriangledt} We say that a triangle $T=xyz$ is \emph{bad} \wrt edge $xy$ if one of the following holds.
    \begin{itemize} 
        \item[(i)] If $xy \in \NORM \cup \FIXED$, and $yz,xz \in \NORM \cup \FIXED$ and $T$ does not satisfy the ultrametric condition.
        \item[(ii)] If $xy \in \NORM \cup \FIXED\cup \UNFIX$ and $yz,xz \in \UNFIX$.
        \item[(iii)] If $xy \in \UNFIX$, $yz,xz \in \NORM \cup \FIXED$, and
        \subitem (iiiA) either the other two edges are distinct, that is, $\dist(x,z) \neq \dist(y,z)$, or
        \subitem  (iiiB) the threshold $\treshold_{u}(x,y) >  \dist(x,z) = \dist(z,y)$.

  \end{itemize} See \Cref{fig:ultrametric-badtr} for some examples of bad triangles of each type.

\end{definition}

Note that if an instance $ (X, \dist, \treshold_{\ell}, \treshold_{u})$ with parameter $k=0$ contains a bad triangle, then this is a \noinstance. The opposite is not true. For example, a triangle $xyz$ with two fixed edges $xz$ and $yz$, say of length $3$, and one unfixed edge $xy$ of length $4$ with upper threshold $\tau_{u}(x,y)=3$, is not a bad triangle. However, its distances do not form an ultrametric.

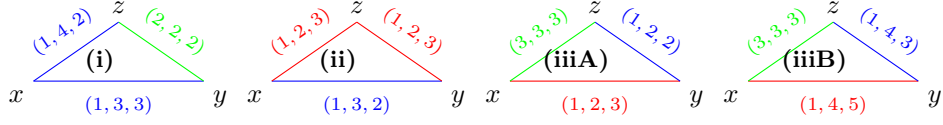
\begin{figure}
\centering
\begin{tikzpicture}[scale=0.9, x=2.5cm] 


\newcommand{\drawtriangle}[2]{
  \begin{scope}[shift={(#1,0)}]
    \coordinate (x) at (0,0);
    \coordinate (y) at (1,0);
    \coordinate (z) at (0.5,0.866);
    \draw[draw=blue]  (x) -- node[midway, sloped, below, yshift=-2pt, text=blue]  {\scriptsize$(1, 3, 3)$} (y);
    \draw[draw=green] (y) -- node[midway, sloped, above, yshift=2pt, text=green] {\scriptsize$(2, 2, 2)$} (z);
    \draw[draw=blue]  (z) -- node[midway, sloped, above, yshift=2pt, text=blue]  {\scriptsize$(1, 4, 2)$} (x);
    \node[below left] at (x) {$x$};
    \node[below right] at (y) {$y$};
    \node[above] at (z) {$z$};
    \coordinate (centroid) at ($ (x)!.333!(y)!.333!(z) $);
    \node at (centroid) {\small\textbf{#2}};
  \end{scope}
}

\newcommand{\drawtriangleii}[2]{
  \begin{scope}[shift={(#1,0)}]
    \coordinate (x) at (0,0);
    \coordinate (y) at (1,0);
    \coordinate (z) at (0.5,0.866);
    \draw[draw=blue] (x) -- node[midway, sloped, below, yshift=-2pt, text=blue] {\scriptsize$(1, 3, 2)$} (y);
    \draw[draw=red]  (y) -- node[midway, sloped, above, yshift=2pt, text=red]  {\scriptsize$(1, 2, 3)$} (z);
    \draw[draw=red]  (z) -- node[midway, sloped, above, yshift=2pt, text=red]  {\scriptsize$(1, 2, 3)$} (x);
    \node[below left] at (x) {$x$};
    \node[below right] at (y) {$y$};
    \node[above] at (z) {$z$};
    \coordinate (centroid) at ($ (x)!.333!(y)!.333!(z) $);
    \node at (centroid) {\small\textbf{#2}};
  \end{scope}
}

\newcommand{\drawtriangleiiiA}[2]{
  \begin{scope}[shift={(#1,0)}]
    \coordinate (x) at (0,0);
    \coordinate (y) at (1,0);
    \coordinate (z) at (0.5,0.866);
    \draw[draw=red]   (x) -- node[midway, sloped, below, yshift=-2pt, text=red]   {\scriptsize$(1, 2, 3)$} (y);
    \draw[draw=blue]  (y) -- node[midway, sloped, above, yshift=2pt, text=blue]  {\scriptsize$(1, 2, 2)$} (z);
    \draw[draw=green] (z) -- node[midway, sloped, above, yshift=2pt, text=green] {\scriptsize$(3, 3, 3)$} (x);
    \node[below left] at (x) {$x$};
    \node[below right] at (y) {$y$};
    \node[above] at (z) {$z$};
    \coordinate (centroid) at ($ (x)!.333!(y)!.333!(z) $);
    \node at (centroid) {\small\textbf{#2}};
  \end{scope}
}

\newcommand{\drawtriangleiiiB}[2]{
  \begin{scope}[shift={(#1,0)}]
    \coordinate (x) at (0,0);
    \coordinate (y) at (1,0);
    \coordinate (z) at (0.5,0.866);
    \draw[draw=red]   (x) -- node[midway, sloped, below, yshift=-2pt, text=red]   {\scriptsize$(1, 4, 5)$} (y);
    \draw[draw=blue]  (y) -- node[midway, sloped, above, yshift=2pt, text=blue]  {\scriptsize$(1, 4, 3)$} (z);
    \draw[draw=green] (z) -- node[midway, sloped, above, yshift=2pt, text=green] {\scriptsize$(3, 3, 3)$} (x);
    \node[below left] at (x) {$x$};
    \node[below right] at (y) {$y$};
    \node[above] at (z) {$z$};
    \coordinate (centroid) at ($ (x)!.333!(y)!.333!(z) $);
    \node at (centroid) {\small\textbf{#2}};
  \end{scope}
}

\drawtriangle{0}{(i)}
\drawtriangleii{1.4}{(ii)}
\drawtriangleiiiA{2.8}{(iiiA)}
\drawtriangleiiiB{4.2}{(iiiB)}

\end{tikzpicture}

\caption{Examples of bad triangles with respect to edge $xy$. The meaning of the 3-tuple labeling an edge of the triangle is the following: the first coordinate is $\treshold_{\ell}$, the second is $\treshold_{u}$, and the third is $\dist$ of that edge. For example, in triangle (i), $\treshold_{\ell}(x,y)=1$, $\treshold_{u}(x,y)=3$, and $\dist(x,y)=3$. In this triangle, blue edges $xy$ and $xz$ are in $\NORM$, and green $yz$ is in $\FIXED$. In triangle (ii), blue edge $xy \in \NORM$ and red edges $xz$, $yz$ are in $\UNFIX$.}
\label{fig:ultrametric-badtr}
\end{figure}

    Let $xyz$ be a bad triangle \wrt edge $xy$.
    If $\dist(x,z) \neq \dist(y,z)$ then there is a unique value we can assign to edge $xy$ to fix this bad triangle whereas if $\dist(x,z) = \dist(y,z)$ then we have an upper bound on the weights that edge $xy$ can attain to satisfy this triangle.
    Indeed we need different reduction rules to handle them and thus we further classify the bad triangles of type $(i)$ and type $(iii)$ into bad isosceles triangles and bad scalene triangles. 

\begin{definition}[Bad isosceles and scalene triangles \wrt edge $xy$] \label{def:badequlatdt}
    Let  $T=xyz$ be a bad triangle \wrt edge $xy$ and edges $yz,zx\in \FIXED \cup \NORM$. We say that triangle $T$ is \emph{isosceles}   
    \begin{itemize}
        \item If $xy \in \NORM \cup \FIXED$ and $\dist(y,z)=\dist(z,x)$, or

        \item If $xy \in \UNFIX$ and the upper threshold $\treshold_{u}(x,y) >  \dist(x,z) = \dist(z,y)$, 
    \end{itemize}
    and \emph{scalene} 
    \begin{itemize}
        \item If $xy \in \NORM \cup \FIXED \cup \UNFIX$ and $\dist(y,z) \neq \dist(z,x)$.

    \end{itemize}
\end{definition}
For example, in \Cref{fig:ultrametric-badtr} triangles (i) and (iiiB) are isosceles.

\medskip
We are ready to proceed with the proof of  \Cref{theorem:kernelconstrultr}, which we restate here. 

\KernelConstrUltr*

\begin{proof} We provide a set of reduction rules that, in polynomial time, transform the given input instance $( (X,  \dist, \treshold_{\ell}, \treshold_{u}), k)$ of \Cuvd into an equivalent one with at most $3k^2 + 2k$ vertices.   

In the reduction rules, we assume that they are applied in the given order. That is, we perform a rule only if the conditions of the preceding rules do not hold. When the rule outputs \noinstance, we construct a trivial \noinstance, say a bad triangle with distances $1,2$, and $3$ and $k=0$.

The first three reduction rules ensure the completeness of the kernelization algorithm.

\begin{leftbar}
    \begin{dtredrule} \label{DTRR:1}
        If there exists an edge $xy$ such that $\treshold_{\ell}(x,y) > \treshold_{u}(x,y) $, output a \noinstance.
    \end{dtredrule}
\end{leftbar}

\begin{claim}
    \Cref{DTRR:1} is sound.
\end{claim}

\begin{proof}
    If  $\treshold_{\ell}(x,y) > \treshold_{u}(x,y)$ for some edge $xy$, then we cannot satisfy the threshold constraint.
\end{proof}

\begin{leftbar}
    \begin{dtredrule} \label{DTRR:2}
        If  
         the total number of $\UNFIX$ edges exceeds $k$, 
        output  a  \noinstance.
    \end{dtredrule}
\end{leftbar}

\begin{claim}
    \Cref{DTRR:2} is sound.
\end{claim}

\begin{proof}
   Since each of the \unfix edges must be changed, we require more than $k$ operations, thereby exceeding the budget $k$. 
\end{proof}

\begin{leftbar}
    \begin{dtredrule} \label{DTRR:0}
        If there exists an edge $xy$ such that $\treshold_{\ell}(x,y) = \treshold_{u}(x,y) $ and $\dist(x,y) \neq \treshold_{\ell}(x,y)$, then set $\dist(x,y) = \treshold_{\ell}(x,y)$ and decrease the parameter $k$ by 1.
    \end{dtredrule}
\end{leftbar}

\begin{claim}
    \Cref{DTRR:0} is sound.
\end{claim}

\begin{proof}
    Since, the only value $\dist(x,y)$ can have is $\treshold_{\ell}(x,y)$ we need to set $\dist(x,y) = \treshold_{\ell}(x,y)$.
\end{proof}

We can now assume that the given instance has at most $k$ \unfix edges and that  for every edge $xy \in \NORM \cup \UNFIX$,  we have $\treshold_{\ell}(x,y) < \treshold_{u}(x,y)$. \Cref{DTRR:3,DTRR:4} ensures that every edge $xy \in \FIXED \cup \UNFIX \cup \NORM$ is a part of at most $k$ bad scalene triangles \wrt $xy$ such that their larger side is equal.

\begin{leftbar}
    \begin{dtredrule} \label{DTRR:3}
        If for some  $xy \in \FIXED$ there are $k+1$ bad triangles \wrt $xy$, then 
        output a \noinstance.  
    \end{dtredrule}
\end{leftbar}

\begin{claim}
    \Cref{DTRR:3} is sound.
\end{claim}

\begin{proof}
    Since we cannot change the distance of a \fixed edge $xy$, we have to perform at least $k+1$ operations in order to change all bad triangles containing $xy$.
\end{proof}

\begin{leftbar}
    \begin{dtredrule} \label{DTRR:4}
        Suppose that 
        there is  a \normal or \unfix edge $xy\in \NORM\cup \UNFIX$ that is contained in at least $k+1$ bad scalene triangles $xyz_1,\dots xyz_{k+1}$ \wrt $xy$ such that the distances $ \max\{\dist(x,z_1),\dist(y,z_1)\} =\cdots= \max\{\dist(x,z_{k+1}),\dist(y,z_{k+1})\} =M$.
        (In words, the lengths of the longest edges, besides $xy$, of these bad scalene triangles are equal to some number $M$.)
        \begin{itemize}
           \item If either $\treshold_{u}(x,y)<M$ or $\treshold_{\ell}(x,y)>M$ output  a \noinstance.
           \item If $\treshold_{\ell}(x,y) \leq M \leq \treshold_{u}(x,y)$, then move $xy $  to $\FIXED$ by setting  $\dist(x,y)=M$ and changing  $\treshold_{\ell}(x,y) = \treshold_{u}(x,y)=M$.
           Decrease the parameter $k$ by $1$. 
       \end{itemize} 
    \end{dtredrule}
\end{leftbar}

\begin{claim}
    \Cref{DTRR:4} is sound.
\end{claim}

\begin{proof}
    Note that since $xyz_i$ are bad for $i\in\{1,\ldots,k+1\}$,  
    we have to modify at least one distance in each triangle.
    
    If either $\treshold_{u}(x,y)<M$ or $\treshold_{\ell}(x,y)>M$ then we cannot set $\dist(x,y)=M$. In this case, to modify these $k+1$ bad scalene triangles containing $xy$, we need at least $k+1$ operations. Hence, we have a \noinstance.  
    
    Suppose that $\treshold_{\ell}(x,y)\leq  M \leq \treshold_{u}(x,y)$. Because of \Cref{DTRR:2}, there is $i\in\{1,\ldots,k+1\}$ such that $xz_i,yz_i\in \NORM\cup\FIXED$. Thus, $\dist(x,y)\neq M$.   
    Then 
    if we do not set $\dist(x,y)=M$, we have to perform more than $k$ operations to fix the bad triangles containing $xy$.
\end{proof}

After exhaustively applying \Cref{DTRR:4}, we know that in the given instance, there are at most $k$ bad scalene triangles \wrt any edge $xy$ such that the larger side is equal in all of them. The next rule ensures that there are at most $2k$ bad scalene triangles \wrt any edge $xy \in \NORM \cup \UNFIX$.  Then by combining \Cref{DTRR:6} with \Cref{DTRR:3}, we conclude that any edge is contained in at most $2k$ bad scalene triangles.

\begin{leftbar}
    \begin{dtredrule} \label{DTRR:6}
        If there are at least $2k$ bad scalene triangles \wrt a \normal or \unfix edge $xy$, output \noinstance. 
    \end{dtredrule}
\end{leftbar}

\begin{claim}
    \Cref{DTRR:6} is sound.
\end{claim}

\begin{proof} 
The only way to fix a bad triangle containing $xy$ by changing $\dist(x,y)$ and not changing the other two edges, is to assign $\dist(x,y)$ the length of the longest edge of this triangle. 
    Because $xy$ belongs to at least $2k$ bad scalene triangles, we need to modify 
    $\dist(x,y)$,    otherwise, we have to spend at least $2k$ modifications on fixing these triangles. 
    Since we cannot apply \Cref{DTRR:4}, each edge $xy \in \NORM \cup \UNFIX$ can be part of at most $k$ bad scalene triangles $xyz_1,\dots xyz_{k}$ \wrt $xy$ such that the distances $ \max\{\dist(x,z_1),\dist(y,z_1)\} =\cdots= \max\{\dist(x,z_{k}),\dist(y,z_{k})\} =M$ and
    $\max\{\dist(x,z_i),\dist(y,z_i)\} \neq \min\{\dist(x,z_i),\dist(y,z_i)\}$ for all $i\in \{1,\dots, k\}$. 
    Thus 
    we cannot fix more than $k$ of these bad triangles by modifying only the distance of $xy$. Therefore, to fix the remaining bad scalene triangles, we still need at least $k$ additional operations, requiring a total of at least $k+1$ operations.
    Hence, this is a \noinstance.
\end{proof}
By \Cref{DTRR:6} and \Cref{DTRR:3}, every edge is contained in at most $2k$ bad scalene triangles. The next reduction rule will be used to bound the number of bad isosceles  triangles 
\wrt an edge, that is, bounding the number of bad isosceles triangle an edge can be part of.

\begin{leftbar}
    \begin{dtredrule} \label{DTRR:7}
        Suppose there are at least $k+1$ bad isosceles triangles 
        $xyz_1,\dots, xyz_{k+1}$ \wrt an edge $xy$, that is, $\dist(x,z_i)=\dist(y,z_i)$ for all $i\in \{1,\dots, k+1\}$, 
        and $xy \in \NORM \cup \UNFIX$. Suppose also that $\dist(x,z_1)\leq \dist(x,z_2)\leq \cdots \leq \dist(x,z_{k+1})$.
        \begin{itemize}
            \item If $\treshold_{\ell}(x,y)> \dist(x,z_{k+1})$, output  a \noinstance.
            \item If $\treshold_{\ell}(x,y) \leq \dist(x,z_{k+1})$, then set $\treshold_{u}(x,y) :=  \dist(x,z_{k+1})$. Do not decrease the parameter $k$.
        \end{itemize}
        
    \end{dtredrule}
\end{leftbar}

\begin{claim}
    \Cref{DTRR:7} is sound.
\end{claim}

\begin{proof}
In every solution, the $\dist(x,y)$ should be at most   $\dist(x,z_{k+1})$ because, otherwise, we need more than $k$ operations to fix at least $k+1$ bad isosceles triangles \wrt edge $xy$ to make an ultrametric. Thus 
    if $\treshold_{\ell}(x,y) > \dist(x,z_{k+1})$, we conclude that we have a \noinstance. 
    Otherwise, because $\dist(x,y)\leq \dist(x,z_{k+1})$  in any solution, we 
    can set  $\treshold_{u}(x,y) :=  \dist(x,z_{k+1})$; note that because $xyz_{k+1}$ is a bad isosceles triangle $\treshold_{u}(x,y)>\dist(x,z_{k+1})$ initially.   
\end{proof}

\begin{claim} \label{DTRR:bound-isos-tri}
    If \Cref{DTRR:7} is no longer applicable, then there are at most $k$ bad isosceles triangle \wrt any edge $xy \in \UNFIX \cup \NORM$.
\end{claim}

\begin{proof}
    Let edge $xy \in \UNFIX\cup \NORM$.  
    If there were at least $k+1$ bad isosceles triangles $xyz_1,\dots, xyz_{k+1}$ \wrt  $xy$, that is $\dist(x,z_i)=\dist(y,z_i)$ for all $i\in \{1,\dots, k+1\}$ and also that $\dist(x,z_1)\leq \dist(x,z_2)\leq \cdots \leq \dist(x,z_{k+1})$, then 
    \Cref{DTRR:7} would be applicable contradicting the assumption.  
\end{proof}

In the previous reduction rules we either changed the weight of an edge, changed its upper threshold or declared that the given instance is a \noinstance. Now we have set premise for \Cref{DTRR:10} which reduces the size of the given instance by removing an ``irrelevant''  vertex  which does not belong to any bad triangle \wrt any edge.

\input{Figures/irrelevant-vertex}

\begin{leftbar}
    \begin{dtredrule}[Irrelevant Vertex Removal] \label{DTRR:10}
        Let $u$ be a vertex  such that no \unfix edge is incident with it and there is no edge $xy$ such that $T = uxy$ is a bad triangle \wrt $xy$.          Delete $u$.  
    \end{dtredrule}
\end{leftbar}

\begin{claim}\label{claim:DTRR:10} 
    \Cref{DTRR:10} is sound.
\end{claim}
 
\begin{proof}
Suppose that in instance $\Xcal=(X, \dist, \treshold_{\ell}, \treshold_{u})$  there is  a vertex  $u\in X$  satisfying the conditions of the reduction rule. 
 
    Let $d_1 < d_2 < \cdots < d_\ell$ be the distinct distances of the edges incident with $u$. We partition the vertex set $X \setminus \{u\}$ into $\ell$ sets $X_1, X_2, \dots, X_\ell$ based on their distances from $u$. More specifically,  
    $
        X_i = \{x\in X \setminus \{u\} \mid \dist(u, x) = d_i \}
    $
    (see \Cref{fig:irrelevant-vertex}).
    
    Since $u$ does not belong to any bad triangle \wrt any edge $xy$, no edge between the sets $X_i$ and $X_j$ for $i \neq j$ can be in $\UNFIX$ because otherwise we would have a bad triangle \wrt edge $xy$ of type $(iii A)$.  
    Moreover, for any two vertex s $x, y$ from different sets $X_i$ and $X_j$ with $i < j$, the distance satisfies $\dist(x, y) = d_j$; otherwise, $uxy$ would form a bad triangle \wrt edge $xy$.  
    Additionally, for any two vertices $x, y$ in the same set $X_i$, since $uxy$ is not a bad triangle \wrt edge $xy$, we have either $\dist(x, y) \leq d_i$ or $xy \in \UNFIX$ with $\treshold_{u}(x,y) \leq d_i$---otherwise we would have a bad triangle \wrt $xy$ of type $(i)$ or type $(iii)$, respectively.

    We are ready to prove the  equivalence:   $( \Xcal,k)$ is a \yesinstance if and only if the reduced instance  $( \Xcal'=(X \backslash \{u\}, \dist, \treshold_{\ell}, \treshold_{u}),k)$ is also a \yesinstance. The only if direction is trivial as we might require fewer changes in the reduced  instance.

    For the reverse direction, if $(\Xcal',k)$ is a \yesinstance, then it can be transformed by at most $k$ operations in an ultrametric.
    Since $X\setminus \{u\}=X_1\cup \cdots \cup X_\ell$, there are nonnegative integers $k_1, \dots, k_\ell$ such that $k_1+k_2+\cdots +k_\ell \leq k$, and each of the subproblems induced by points $X_i$ could be transformed into an ultrametric  $(X_i, \dist_i)$ by at most $k_i$ operations. We claim that the union of all these operations would result in an ultrametric within $k$ operations from $( \Xcal,k)$. 
 
    For $x,y \in X_i$, since $uxy$ is not a bad triangle \wrt $xy$, we have that if $xy \in \UNFIX$ then $\tau_{u}(x,y) \leq d_i$.
    Also, if $xy \in \FIXED$,  we have  that $\dist(x,y) \leq d_i$. However, it could happen that for $xy \in \NORM$, we could have $\treshold_{u}(x,y) > d_i$ but then $\treshold_{\ell}(x,y) \leq \dist(x,y) \leq d_i$ as $xy \in \NORM$.
    Thus, there could exist a solution for $X_i$ such that for some edges their weight exceeds $d_i$.
    Additionally, observe that only edges in $\NORM$ can have their weight exceeding $d_i$.
    By setting $\lambda = d_i$ in \Cref{lemma:truncation} we can convert a solution where the weights of edges exceed $d_i$ to a solution where the weights of edges are at most $d_i$ in the same budget.

       Thus, we can assume that in every  $(X_i, \dist_i)$, the maximum value of $\dist_i$ does not exceed $d_i$. We define the following distance $\dist'$ on $X$. 
    For all $x\in X$, $\dist'(u,x)=\dist(u,x)$. For every $x$ and $y$ from distinct $X_i$ and $X_j$, $\dist'(y,x)=\dist(y,x)$.
    For every $x,y\in X_i$, $1\leq i\leq \ell$, $\dist'(y,x)=\dist_i(y,x)$.
    Thus $(X,\dist')$ is obtained from $ \Xcal$ by at most  $k_1+k_2+\cdots +k_\ell \leq k$ operations.  What remains is to check that  $(X,\dist')$ is an ultrametric. 

    For vertices   $x,y,z\in X$ we consider the following cases. 
    \begin{itemize}
        \item  $x,y,z\in X_i$ for some $i$.  Since  $(X_i, \dist_i)$ is ultrametric,  $x,y,z$ satisfy the ultrametric condition in $\dist'$ too.
        \item  $x\in X_i$, $y\in X_j$, $z\in X_k$, $i<j<k$. Since we did not change the distances between these vertices, we have $\dist'(x,y) = d_j$ and $\dist'(x,z)= \dist'(y,z) = d_k > d_j$ satisfying the ultrametric conditions. 
        \item  $x,y\in X_i$, $z\in X_j$, $i<j$ or $i>j$. If $i>j$ then $\dist'(x,z) = \dist'(y,z) = d_i$ and $\dist'(x,y) \leq d_i$, satisfying the ultrametric condition.   If $i<j$ then $\dist'(x,z) = \dist'(y,z) = d_j$ and $\dist'(x,y) \leq d_i < d_j$, satisfying the ultrametric condition. In either case, we do not violate the ultrametric condition.      
        \item  $x \in X_i$, $y\in X_j$, $i<j$, and $z=u$. We have $\dist'(x,u) = d_i$ and $\dist'(y,u) = d_j$ Since we did not modify $xy$, we have that $\dist'(x,y) = d_j$ and $d_i < d_j$, this triangle satisfies the ultrametric condition.
        \item  $x,y \in X_i$, $z=u$. Since $\dist'_i \leq d_i$, and $\dist'(x,u) = \dist'(y,u) = d_i$, the triangle $xyu$ satisfy the ultrametric condition in $\dist'$. (Here $\dist'_i$ is $\dist$ restrict to points in $X_i$)
    \end{itemize}

    This concludes the proof that $(X,\dist')$ is an ultrametric and the proof of \Cref{claim:DTRR:10}.
\end{proof}

We are ready to conclude the proof of the theorem. We will bound the number of different types of bad triangles according to \Cref{def:badtriangledt} and consequently bound the number of vertices in the reduced instance.

\begin{leftbar}
    \begin{dtredrule} \label{DTRR:final}
     If an irreducible instance (that is, none of the previous reduction rules applies to this instance) has more than $3k^2 + 2k$ vertices, then output that this is  a \noinstance. 
    \end{dtredrule}
\end{leftbar}

\begin{claim}
    \Cref{DTRR:final} is sound.
\end{claim}

 \begin{proof}

Let $(\Xcal= (X, \dist, \treshold_{\ell}, \treshold_{u}), k)$ be an irreducible  \yesinstance instance of \Cuvd.   Then there is a set  $A$  of at most $k$ edges whose distance modification transforms $\Xcal$  into an ultrametric. Note that by \Cref{DTRR:0}, $A$ contains all edges of $\UNFIX$ and, possible,  some edges of $\NORM$. Also, by \Cref{DTRR:1} we know that for each edge $xy$, we have $\treshold_{\ell} \leq \tau_{u}$.

Let $xy \in \NORM$ be an edge of $A$. By \Cref{DTRR:6}, there are at most $2k$ bad scalene triangles \wrt edge $xy$ and, by \Cref{DTRR:bound-isos-tri}, there are at most $k$ bad isosceles triangles \wrt $xy$. Thus, there are at most $3k$ bad triangles of type $(i)$ \wrt every \normal edge in $A$. Thus, the number of vertices in $X$ which can be part of a bad triangle of type $(i)$ do not exceed $|A \cap \NORM| \cdot (3k+2)$.

Let $xy \in \UNFIX$ be an edge of $A$, by \Cref{DTRR:6} there are at most $2k$ bad scalene triangles \wrt $xy$ and, by \Cref{DTRR:bound-isos-tri} there are at most $k$ bad isosceles triangles \wrt $xy$. Thus, there are at most $3k$ bad triangles of type $(iii)$ \wrt every \unfix edge in $A$. 
Hence the number of vertices in $X$ which belong to bad triangles of type $(iii)$ is at most $|\UNFIX|\cdot(3k+2)$.

For any bad triangle of type $(ii)$, each vertex  of such a triangle is incident with an \unfix edge. Thus, the number of such vertices can be at most $2k$ by \Cref{DTRR:2}. Note that these vertices are already taken into account while bounding the number of vertices in $X$ belonging to bad triangles of type $(iii)$.

Finally, by \Cref{DTRR:10}, each vertex  of $X$  belongs to a bad triangle. However, as we have seen, the number of vertices belonging to bad triangles is at most 
\[
|A\cap \NORM| \cdot (3k+2) +|\UNFIX|\cdot (3k+2)\leq 3k^2+2k.
\]
We conclude that every irreducible  \yesinstance instance of \Cuvd contains at most  $3k^2+2k$ vertices.
\end{proof}

Note that, at this point, we have an instance of \Cuvd{} with a bounded number of vertices. However, the values of the distances and thresholds may still be unbounded in terms of the parameter $k$. Therefore, we do not yet have a polynomial kernel.

Our next goal is to scale down the values of the distances and thresholds so that they are bounded by some function of $k$.

We sort all values in the set \(\Dvalues \cup \Treshold_{u} \cup \Treshold_{\ell} \) in increasing order:
\[
d_1 < d_2 < \cdots < d_L.
\]
We then define a scaling function \( f : \Dvalues \cup \Treshold_{u} \cup \Treshold_{\ell} \rightarrow \{1, 2, \ldots, L\} \) by setting \( f(d_i) = i \) for each $i$.

We now argue that scaling the distances and thresholds in $\Xcal$ using the function $f$ results in an equivalent instance. More formally,

\begin{claim}
    \label{claim:scaling-weights}
    $ \Xcal = (X, \dist, \treshold_{\ell}, \treshold_{u})$  is a \yesinstance if and only if the instance $ {\Xcal}_f$ obtained by replacing $\dist(x,y)$ with $f(\dist(x,y))$, 
    $\treshold_{u}(x,y)$ with $f(\treshold_{u}(x,y))$, and $\treshold_{\ell}(x,y)$ with $f(\treshold_{\ell}(x,y))$,
    is a \yesinstance.
\end{claim}

\begin{proof}
    Observe that the function $f$ is a strictly increasing function.
    We argue that a triangle satisfies the ultrametric condition in $\Xcal$ if and only if it satisfies the ultrametric condition in ${\Xcal}_f$. For a triangle $T=xyz$ in $\Xcal$ satisfying the ultrametric condition, the following two cases can occur:  
    
 \smallskip\noindent\textbf{Case 1.} $\dist(x,y) = \dist(y,z) = \dist(x,z)$. 
    Because $f$ is a strictly increasing, we have that $x=y$ if and only if $f(x) = f(y)$. Therefore, in ${\Xcal}_f$ we have $f(\dist(x,y)) = f(\dist(y,z)) = f(\dist(x,z))$ satisfying the ultrametric condition. 
    
    \smallskip\noindent\textbf{Case 2.} $\dist(x,y) = \dist(y,z) > \dist(x,z)$.
    Since $f$ is a strictly increasing function, we have that $x>y$ if and only if $f(x) > f(y)$. Therefore, in ${\Xcal}_f$ we have $f(\dist(x,y)) = f(\dist(y,z)) > f(\dist(x,z))$ satisfying the ultrametric condition. 
   Similarly,  the reverse direction also follows from the monotonicity of $f$.
   
  Also for any $x,y\in X$,  $\treshold_{\ell}(x,y)\leq \dist(x,y)\leq \treshold_{u}(x,y)$ if and only if  $f(\treshold_{\ell}(x,y))\leq f(\dist(x,y))\leq f(\treshold_{u}(x,y))$.

    By \Cref{corollary:tight-solution}, we know that if a given instance $\Xcal$ with corresponding distance matrix $D$ is a \yesinstance then it can be transformed into ultrametric space $\Xcal' = \{X,\dist'\}$
    by performing at most $k$ changes such that the values of   $ \dist'$ are only in $\Dvalues \cup \Treshold_{u} \cup \Treshold_{\ell}$. 

    Let $D_f$ be the distance matric corresponding to ${\Xcal}_f$ and $D_f'$ be the distance matrix obtained by applying the function $f$ on every entry of distance matrix $D'$ of $\Xcal'$.  Since $D'$  is an ultrametric, we have that $D_f'$ is also an ultrametric. Finally, 
 matrices $D$ and $D'$ differ in $k$ entries if and only if  $D_f$ and $D_f'$  differ in $k$ entries. It means that  $\Xcal$ an \yesinstance if and only if $\Xcal_f$ is.
\end{proof}

We have shown that by applying  reduction rules exhaustively, for a given instance $(\Xcal= (X, \dist, \treshold_{\ell}, \treshold_{u}), k)$   of \Cuvd, we either construct an equivalent instance on at most  $3k^2+2k$ vertices or correctly conclude that  $\Xcal$ is a \noinstance. 
Thus, the number of vertices in an instance is atmost $3k^2+2k$ making the total number of edges at most $(3k^2+2k)^2 = \Ocal(k^4)$.  By applying the scaling procedure of \Cref{claim:scaling-weights}, we obtain an equivalent instance with maximum distance at most  $L =|\Dvalues \cup \Treshold_{u} \cup \Treshold_{\ell}|=\Ocal(k^4)$. 
Thus, for each pair of vertices we need $\Ocal(\log k)$ bits to encode the distance between them. Hence we have kernel of size $\Ocal(k^4\log{k})$.

Because the total number of triangles is $\binom{|X|}{3}$, we have that the reduction rules can be executed in polynomial time.
This concludes the proof of the theorem.
\end{proof}


    Let us remark that \Cref{theorem:kernelconstrultr} yields an algorithm solving \Cuvd in time $k^{\cO(k)}\cdot |X|^{\cO(1)}$. 

    Indeed, after constructing an equivalent instance $(\Xcal= (X, \dist, \treshold_{\ell}, \treshold_{u}), k)$   of \Cuvd on $3k^2+2k$ vertices, we guess at most $k$ edges that has to be changed in order to turn $\Xcal$ into ultrametric. Of course, each of these guesses contains all edges of $\UNFIX$ by default. Since $\Xcal$ contains $\cO(k^4)$ edges, we make  $k^{\cO(k)}$ guesses. For each of the guesses, 
    it is possible to identify in polynomial time (see  \Cref{lemma:poly-if-known-edges}) whether it is possible to transform $\Xcal$ into an ultrametric by changing the distances of the guessed edges. 
 
In the next section, we  provide a faster branching algorithm.

\section{Single-exponential algorithm for \cuvd}\label{sec:branching}



    In this section, we present a single-exponential FPT algorithm for \Cuvd{} with running time
    $
      \Ocal\big(9^k \cdot k \cdot |X|^{2}\big).
    $
    The algorithm follows a \emph{divisive} paradigm.  At each recursive step, we solve a suitable variant of the \textsc{Cluster Editing} problem.  Unlike the standard branching scheme for \textsc{Cluster Editing}, however, the weights of certain edges may be altered multiple times during the computation.  This peculiarity complicates both the branching rules and the analysis of the overall running time.

    Before proceeding with the algorithm, we need a few auxiliary results. 
    For a given instance $(\Xcal = (X, \dist, \treshold_{\ell}, \treshold_{u}),\ k)$ of \Cuvd and the corresponding distance matrix $D$, we sort the distinct entries in $D$ together with the values from $\Treshold_{\ell}$ and $\Treshold_{u}$ as $d_1 < d_2 < \cdots < d_L$, where $L = |\Dvalues \cup \Treshold_{\ell} \cup \Treshold_{u}| = \mathcal{O}(|X|^2)$.
    By \Cref{corollary:tight-solution}, we can restrict our search to solutions with distances taking values only from $\Dvalues \cup \Treshold_{\ell} \cup \Treshold_{u}$.
    In the algorithm, we consider a set of \emph{marked} edges $M$, which is the set of edges whose weight we have decided to change but are unsure of what weights to assign.
    We will explain further in detail how we use the set of marked edges in our algorithm.
    We refer to each index $i \in \{1, \dots, L\}$ as a \emph{level}.
   
    We proceed in the algorithm from the topmost level $L$ down to smaller levels. 
    At each level, we solve the \textsc{Cluster Editing} problem on the corresponding auxiliary graph which we define as follows.

    \begin{definition} \label{auxiliary-graph}
        (Auxiliary graph $G_i(\dist,M)$) Given a distance function $\dist$ on $X$ and a set of marked edges $M$ with their weights in $\Dvalues \cup \Treshold_{\ell} \cup \Treshold_{u}$, the auxiliary graph $G_i(\dist,M)$ with respect to level $i\in\{1,\dots, L\}$ is a graph on the vertex set $X$. The edge set of $G_i(\dist,M)$ is 
        \[ \{xy \mid \dist(x,y) < d_i, x\neq y, x,y \in X \} \cup M.\]
    \end{definition}
    For the level $L$, $\dist$ is the input distance function and $M$ is the set of \unfix edges $\UNFIX$ (edges whose current weight do not satisfy the threshold constraints). However, in the recursive calls of the algorithm, the distance function will be changing. 
    We may write $G_i$ instead of $G_i(\dist,M)$ if the distance function is clear from the context.

    Let us start with the graph $G_L=G_L(\dist, \UNFIX)$. Two observations are in order.
    First, if $\mathcal{X}$ is a \yesinstance{}, then by \Cref{corollary:tight-solution}, there exists a solution in which the maximum distance between vertices does not exceed $d_L$.
    Also, the solution ultrametric space $(X,\dist')$ should not contain any marked edges.
    Second, the auxiliary graph $G'_L=G_L(\dist',\emptyset)$ corresponding to the resulting ultrametric $(X, \dist')$ must be a cluster graph, that is, a disjoint union of cliques. Indeed, if $G'_L$ contains a connected component that is not a clique, then it must contain an induced path $P_3$ on three vertices. For example, suppose $x, y, z$ form such a path, with $xy, yz \in E(G'_L)$ and $xz \notin E(G'_L)$. This implies that $\dist'(x, y), \dist'(y, z) < d_L$ while $\dist'(x, z) = d_L$, which contradicts the ultrametric property of $(X, \dist')$.
    Thus, if $\mathcal{X}$ is a \yesinstance{}, it must admit a solution that transforms $G_L$ into a cluster graph by at most $k$ edits. This allows us to branch on every edge and non-edge of each $P_3$ in $G_L$: we either increase the distance of an edge to $d_L$ (corresponding to deleting an edge in $G_L$), or decrease the distance (corresponding to adding an edge in $G_L$). This is exactly the strategy used in a branching algorithm for the \textsc{Cluster Editing} problem on graph $G_i$.

    However, solving \textsc{Cluster Editing} at a lower level introduces additional complexity. When we increase the distance of an edge to $d_L$, we can safely reduce the budget $k$ by one, and we will not modify this edge again. But when we decrease the distance of an edge, we do not yet know how far to decrease it at the current level. To be safe, we reduce it to the next level $d _{L-1}$.
    The distance of such an edge may be decreased multiple times in recursive calls to the branching procedure at lower levels. To control the parameterized running time, we must carefully track the progress of the branching process.

    Recall that $G[X]$ is used to denote the complete graph induced by distance space $(X,\dist)$, that is, the edges of $G[X]$ are the pairs of points of $X$ and the edge  weights are the distances between the points. 
    During the algorithm's execution, for each level $i$, we partition the  edges of $G[X]$ into the following three sets.
    
    \begin{itemize}
        \item $D_i$: Set of edges whose distance is equal to $d_i$.
        \item $D_{<i}$: Set of edges whose distance is strictly less than $d_i$.
      \item $M$: The set of \emph{marked edges}. These are edges whose weights were assigned to be modified at earlier (higher) levels. After moving to the next  level $i-1$, a marked edge may either be fixed by assigning it the weight $d_i$, or it may remain marked.
    \end{itemize} 

    For each level $i$, our algorithm branches on a bad triangle. Given the new classification of edges, we have a new definition for ``bad'' triangles.

    \begin{definition} \label{bad-triangles-level} 
        (Violated triangle with respect to level $i$.) We call a triangle $T = xyz$ a \emph{violated triangle with respect to level $i$}, if $T$ has exactly one edge in $D_i$ and the other two edges are from  $M\cup D_{<i}$.
    \end{definition}

    Now, we prove that if there are no violated triangles with respect to each level $j\in\{i,\dots,L\}$,  then $G_i$ is a cluster graph. Moreover, if there is a  violated triangle for some lower level
      $\ell \leq i-1$, then all three vertices of this violated triangle belong to one clique of the cluster graph $G_i$.
    
    \begin{lemma}
    \label{claim:G-i-cluster}
        If the given instance $(\Xcal=(X, \dist, \treshold_{\ell}, \treshold_{u}), k)$ does not have any violated triangles with respect to all levels $j\in\{i,\dots,L\}$, then the auxiliary graph $G_i$ is a cluster graph with maximal cliques $\Ccal_{i} = \{C_1, C_2, \dots,C_p\}$. Moreover, if there exists a violated triangle $T = xyz$ with respect to level  $\ell \leq i-1$, then there exists a clique $C_m\in \Ccal_{i}$ such that $x,y,z \in C_m$.
    \end{lemma}

    \begin{proof}
         We prove the claim  by induction on the distance levels. As we discussed already, for the base case  level $i = L$, if we do not have a violated triangle with respect to level $L$ then in $G_L$ there is no induced $P_3$, making it a cluster graph. Let $\Ccal_{L}= \{C_1, C_2, \dots,C_p\}$ be connected components (which are cliques) of $G_L$. Let  $T = xyz$ be a violated triangle with respect to level  $\ell \leq i-1$. By the definition of the violated triangle, one edge of $T$ is of length $d_\ell$ and the other two edges are from $M\cup D_{<\ell}$.
         On the other hand, 
         if one of the vertices, say $x$,  is not in the same clique of $\Ccal_{L}$ with $y$ and $z$,  then $d_L=\dist(x,y) =\dist(x,z)$ and $T$ is not a violated triangle  with respect to level  $\ell$. Thus, a violated triangle with respect to level  $\ell \leq i-1$ could occur only when $x,y,z \in C_m$ for some $C_m \in \Ccal_{L}$. 
 
         For the inductive case, when $i < L$, by the inductive hypothesis we know that $G_{i+1}$ is a cluster graph.
         Suppose, for contradiction, that $G_i$ is not a cluster graph. Then it contains a $P_3 = xyz$ such that
         $xy, yz \in E(G_i)$ and $xz \notin E(G_i)$. By the construction of the auxiliary graph, this implies that
         $\dist(x, z) \geq d_i$ and $xy, yz \in D_i \cup M$. Since $G_{i+1}$ is a cluster graph, the triple $xyz$ cannot induce a $P_3$ in $G_{i+1}$, which means that $\dist(x, z) = d_i$. But then the vertices $x, y, z$ form a violated triangle with respect to level $i$, contradicting our assumption.
   
         Let $T = xyz$ be a violated triangle with respect to level $\ell < i$. By the induction assumption, $x, y, z$ belong to the same clique in $G_{i+1}$, and hence the distances between them are all less than $d_{i+1}$. Now, suppose that in $G_i$, vertex $x$ belongs to a different clique than $y$ and $z$. Then $\dist(x, y) = \dist(x, z) = d_i$, and thus $T$ is not a violated triangle with respect to level $\ell$.
    \end{proof} 


    The branching algorithm converts instance of \Cuvd $(\Xcal=(X, \dist, \treshold_{\ell}, \treshold_{u}), k)$ into an ultrametric if and only if a solution can be obtained by modifying at most $k$ edges (pair of vertices) in $X$. 
    In order to do so, the algorithm either changes the weights of some edges or marks some edges in order to fix violated triangles with respect to level $i$.
    However, even if the instance does not contain any violated triangles, it still has some marked edges that need to be assigned some weights. 
    This instance with marked edges can also be viewed as an instance of \Cuvd with a set of edges, and we are only allowed to modify this set of edges.
    Recall that \Cref{lemma:poly-if-known-edges} gives a polynomial-time algorithm for checking whether a given instance of \Cuvd and a set of at most $k$ edges could be turned into a solution for \Cuvd. 

    We are ready to proceed with the main result of this section, the proof of \Cref{theorem:fpt-cuvd}. We restate it here.  

\FPTCuvd*


 
\begin{proof}
    As discussed earlier, our algorithm solves the \textsc{Cluster Editing} instance for each level $i$, such that $1 \leq i \leq L$, starting with the highest level $L$.
    We do this by fixing a violated triangle \wrt a level $i$.
    For that, we need to change the weight of at least one of the edges in the violated triangle. 
    Whenever we increase the weight of an edge $xy \in D_{<i}$, whose weight was strictly less than $d_i$ to $d_i$ we do it directly that is we set $\dist(x,y) = d_i$.
    But while decreasing the weight of an edge $xy$ whose weight was $d_i$ to something less than $d_i$ we do not know what value to change it to, and hence we introduced the intermediate step of marking edges.
    To track the progress of our branching, we set the measure $\mu = 2k$ and we have the following justification.
    The distance of an edge may either increase or decrease. 
    Whenever we increase the weight of an edge, i.e. when we change an edge in $ D_{<i} \rightarrow D_i$ we do not have any intermediate step and thus we decrease the measure by 2, whenever we decrease the weight of an edge we first mark it, i.e. we change an edge in $D_{i} \rightarrow M$, we decrease the measure by $1$, and then we assign a weight to marked edge, i.e. when we change and edge in $M \rightarrow D_{i}$ we decrease the measure by $1$. 
    Thus, to change the weight of any edge from one value to another, we spend $2$ budget, and hence setting $\mu = 2k$ is justified. 
    
    From \Cref{corollary:tight-solution} we can infer that there exists a solution $D'$ which uses entries only from $\Dcal \cup \ \Treshold_{u} \cup \Treshold_{\ell}$. We also know that $|\Dvalues \cup \Treshold_{u} \cup \Treshold_{\ell}| \leq 3|X|^2$ and thus we only have a polynomial many entries to choose our distance values from.
    
    Now we are ready to describe the branching algorithm and prove its correctness.

    \medskip\noindent\textbf{Description of the branching algorithm.} Let the given   instance of \Cuvd be $\Xcal=(X, \dist, \treshold_{\ell}, \treshold_{u}), k)$ with distance matrix $D$. 
    We sort the values of $\Dvalues \cup \Treshold_{u} \cup \Treshold_{\ell}$ as $d_1 < d_2 \cdots < d_L$.  
    First, 
    the distances of all unfit edges  
    $xy \in \UNFIX$ (the edges whose current distance values do not belong to their threshold intervals) have to be changed. So we mark these edges and reduce the budget $\mu = \mu - |\Ucal|$.   
    Our branching instance takes as input 
    \begin{itemize}
    \item The instance of \Cuvd $\Xcal=(X, \dist, \treshold_{\ell}, \treshold_{u}, k)$, 
    \item The set of marked edges $M \subseteq X \times X$, 
    \item The remaining budget  $\mu$,  and the current distance level $i$.
    \end{itemize}
    
    We initialize the branching algorithm with arguments $(\Xcal, \Ucal, \mu=2k-|\Ucal|, i=L)$. In the process of the algorithm, when we reach level $i$, we assume that all violated triangles at higher levels $\{i+1, \dots,L\}$ are already fixed. 
    Additionally, we create the auxiliary graph $G_L$ and make changes to it as we change the weights of the edges.
 
    The algorithm outputs either a \noinstance or converts the given instance $\Xcal$ of \Cuvd into an ultrametric by modifying at most $k$ edges. 
 
    In the process of branching, we also decrease the measure $\mu$ (see below how it is done). 
    For each recursive call, we first perform some pre-processing for the marked edges.
    If, in some branch of the algorithm, we reach a level $i$ and there exists some marked edges whose lower threshold as $d_i$, the only option is to set the value of this edge to $d_i$ and decrease the measure by $1$ (Line $3$ od \Cref{alg:Algo 1}).
    Otherwise, if there exists a marked edge whose lower threshold is greater than $d_i$ then this branch cannot lead to a solution satisfying the constraints, so we abandon this branch(Line $8$ of \Cref{alg:Algo 1}).

    After handling the pre-processing with the marked edges if, in some branch of the algorithm, we reach a level $i$ that contains a violated triangle but the budget has exhausted, i.e., $\mu=0$, this branch cannot lead to a solution, so we abandon it (Line $10$ of \Cref{alg:Algo 1}).
    
    If there exists a violated triangle $T = xyz$ with respect to level $i$, we branch according to the different cases, depending on the type of violated triangle, which are discussed below. 
    If there are no violated triangles, then the auxiliary graph induced by level $i$ is a cluster graph, and in this case, we call the branching on level $i-1$ with this finer partition of $X$ and correspondingly create the auxiliary graph $G_{i-1}$ (Line $17$ of \Cref{alg:Algo 1}). 
    
    As the result, either all branches are abandoned, in which case we have a \noinstance. 
    Otherwise in some branches we reach level $i = 0$, which means that there are no violated triangles with respect to any level and there are no marked edges as each edge has its lower threshold at least $d_1$.
    If the value of $\mu$ is less than $0$, this means we used more budget than allowed and thus, we return \textsc{False} in this case (Line $13$ of \Cref{alg:Algo 1}). 
    Otherwise, we have no violated triangles with respect to any level and all the edges satisfy the threshold constraints and we have an ultrametric
    (Line $16$ of \Cref{alg:Algo 1}).
     
    \begin{algorithm}[h]
        \label{alg:Algo 1}
        \caption{Algorithm  $\mathcal{A} (\Xcal, M, \mu,i )$}
        \KwIn{Instance of \Cuvd
        $\Xcal=(X, \dist, \treshold_{\ell}, \treshold_{u})$, $M \subseteq X \times X$, 
        $ \mu \in \Zcal$,
        $ i \in \{1,\dots, L\}$}
        \KwOut{Either a \noinstance or converts $D$ to ultrametric}
        
        \ForEach{$xy \in M$}{
            \If{$\treshold_{\ell}(x,y) = d_i$}{
                \eIf{$\mu \geq 1$}{
                    $\dist(x,y) \gets d_i$\;
                    $\mu \gets \mu - 1$\;
                }{
                    \Return False \;
                }
            }
            \If{$\treshold_{\ell}(x,y) > d_i$}{
                \Return False\;
            }
        }
        
        \If{$\mu \leq 0$ and ($\exists$ a violated triangle or $|M|>0$)}{
            \Return False\;
        }

        \If{$i = 0$}{
            \eIf{$\mu < 0$}{
                \Return False\;
            }{
                \Return True and  return the Ultrametric \;
                
            }
            
        }
        
        \If{ $\nexists$ a violated triangle \wrt level $i$}{
            \Return $\mathcal{A} (\Xcal, M, \mu,i-1)$ \;
        }
        
        \If{$\exists$ a violated triangle $T=xyz$ \wrt level $i$}{
            make recursive calls to \Cref{alg:Algo 1} according to cases in \textbf{\hyperref[framed:important]{Branching Cases Subroutine}} and \Return \textit{OR} of all recursive calls\;
        }
    \end{algorithm}
    
    Recall \Cref{bad-triangles-level} for  violated triangles with respect to   level $i$: 
    A triangle is bad with respect to level~$i$ if it contains exactly one edge in $D_i$, 
    while the remaining two edges belong either to $D_{<i}$ or to $M$. 
    Note that such a triangle lies entirely within a partition induced by $G_{i+1}$. 
    Based on this, we distinguish the following three types of violated triangles $T = xyz$.
    \begin{enumerate}
    \item $xy \in D_i, yz \in D_{<i}, xz \in M$.
       \item $xy \in D_i, yz \in D_{<i}, xz \in D_{<i}$.
        \item $xy \in D_i, yz \in M, xz \in M$.
    \end{enumerate}

    \begin{claim} \label{claim:only-decrease-d-i}
        At a given (distance) level $i$, we can only decrease the weights of edges $e \in D_i$.
    \end{claim}

    \begin{proof}
        We prove using induction on level $i$. For the base case, when level $i = L$, by \Cref{corollary:tight-solution} we know that if our given instance of \Cuvd is an \yesinstance then we use entries only from $D \cup \Treshold_{u} \cup \Treshold_{\ell}$. So, it is safe to only decrease edges in $D_L$.
        For the inductive case, when $i < L$, by the inductive hypothesis for all levels $i<i'\leq L$, it was safe to only decrease edges in $D_{i'}$. Since we are at level $i$, then by \Cref{claim:G-i-cluster} the auxiliary graph $G_{i+1}$ is a cluster graph which induces a partition on $X$. All the edges in these partitions are either marked or have a weight strictly less than $d_{i+1}$. Hence, the only alternative for any edge $e \in D_i$ is to reduce it. This concludes the proof.
    \end{proof}
    
To fix a violated triangle, we must change the weight of at least one of its edges.
We do the following changes based on the type of edge in a violated triangle we wish to change:
\begin{itemize}
    \item If edge $xy \in D_i$, which means $\dist(x,y) = d_i$ then by \Cref{claim:only-decrease-d-i} we can only decrease the weight of such an edge and thus we mark this edge as the exact weight to assign this edge is not yet determined. 
    \item If edge $xy \in D_{<i}$, which means $\dist(x,y) < d_i$. Since a violated triangle at level $i$ contains exactly one edge in $D_i$, we can only increase the weight of such an edge and thus we set $\dist(x,y) = d_i$.
    \item If edge $xy \in M$, then the only change we can do is setting its weight $\dist(x,y) = d_i$.
\end{itemize}

    Consequently, our branching is exhaustive. Based on these operations, we now give a detailed description of the branching algorithm for all the aforementioned cases of violated triangles. 
    \textit{Note that we set the weight of an edge $xy$ to a value $d_i$ at level $i$ only if $\treshold_{\ell}(x,y) \leq d_i \leq \treshold_{u}(x,y)$ else we do not branch by changing weight of such an edge.}
    Additionally, observe that by \Cref{bad-triangles-level} every violated triangle has an edge in $D_i$ and thus there will be at least one valid edge to branch on.

    We say that an instance of branching is bad if we reach $\mu=0$ and we still have violated triangles left, or we have $|M| \neq 0$ and hence we can return \noinstance.

    If there are no violated triangles with respect to level $i$, then by \Cref{claim:G-i-cluster} 
    the auxiliary graph $G_i$ is a cluster graph. 
    Moreover, \Cref{claim:G-i-cluster} also implies that if a violated triangle exists, 
    then it must lie within a partition induced by $G_i$. 
    Hence, it is safe to proceed with recursion on level $i-1$.

    We argue the correctness of the branching procedure.
    We give the proof by induction on the measure $\mu$.
    For the base case, if $\mu = 0$ and we still have violated triangles left, or we have $|M|>0$ then we return \textsc{False}. 
    For the inductive case, that is, $\mu > 0$. 
    If the instance is an ultrametric then we are done.
    Else, there either exists at least one violated triangle with respect to some level $i$ or the instance contains some marked edges.
    In the case where the instance contains some marked edges notice that in our algorithm if we reach some level $i$ such that the lower threshold of some marked edge $xy$ is $d_i$, the only choice for our algorithm is to set the value of $\dist(xy)$ to $d_i$ and reduce the measure $\mu$ by $1$.
    If the instance contains a violated triangle with respect to some level $i$, then we know that the weights of at least one edge must be changed, which reduces the budget by at least one in every branch. 
    By the inductive hypothesis, we assume that the algorithm correctly returns either an ultrametric or 
    a \noinstance for all values $\mu' < \mu$. 
    Since in every possible choice the measure $\mu$ decreases, and our branching is exhaustive, the correctness of the algorithm follows by induction.

    \begin{mdframed}
        \phantomsection
        \label{framed:important}
        \medskip\noindent\textbf{Branching Cases Subroutine.} The branching subroutine takes as input a violated triangle $T =xyz$ and calls \Cref{alg:Algo 1} recursively based on the type of edges $T$ has.
        
        \begin{itemize}
        
            \item \textbf{Case 1. The violated triangle $T=xyz$ is such that $xy \in D_i, yz \in M, xz \in D_{<i}$}. 
                  In this case, we branch in the following three cases 
                    \begin{itemize}
                        \item  $xy \in D_i  \rightarrow xy \in M $ and decrease $\mu \rightarrow  \mu-1$, \\  
                        (this notation means that we add edge $xy$ whose value was $d_i$ to set of marked edges $M$ and decrease the measure $\mu$ by $1$ or more formally set recursively call the \Cref{alg:Algo 1} on the instance $((\Xcal = (X,\dist,\treshold_{\ell},\treshold_{u}),M \cup\{xy\}, \mu-1,i)$), 
                       
                        \item  $yz \in M  \rightarrow yz \in D_i $  and  $\mu \rightarrow  \mu-1$, or \\
                        (this means that we set the weight of a marked edge $yz$ to $d_i$ and decrease the measure $\mu$ by $1$ 
                        or more formally set $\dist(y,z) = d_i$ and recursively call the \Cref{alg:Algo 1} on the instance $((\Xcal = (X,\dist,\treshold_{\ell},\treshold_{u}),M \backslash\{yz\}, \mu-1,i)$), 
                        
                        \item  $xz \in D_{<i}  \rightarrow xz \in D_i $  and   $\mu \rightarrow  \mu-2$. \\
                        (this means that we set the weight of edge $xz$ whose weight was less than $d_i$ to $d_i$ and decrease the measure $\mu$ by $2$ 
                        or more formally set $\dist(x,z) = d_i$ recursively call the \Cref{alg:Algo 1} on the instance $((\Xcal = (X,\dist,\treshold_{\ell},\treshold_{u}),M , \mu-2,i)$).
                    \end{itemize} 
                In each of these cases, the measure $\mu$ drops by at least one.
                The notation follows similarly for the remaining cases of violated triangles. 
    
            \item \textbf{Case 2.} \textbf{The violated triangle $T=xyz$ is such that  $xy \in D_i, yz \in D_{<i}, xz \in D_{<i}$}.
                In this case, we branch in the following three cases 
                \begin{itemize}
                    \item  $xy \in D_i  \rightarrow xy \in M $ and decrease $\mu \rightarrow  \mu-1$, 
                    \item  $yz \in D_{<i}  \rightarrow yz \in D_i $  and  $\mu \rightarrow  \mu-2$, 
                    \item  $xz \in D_{<i}  \rightarrow xz \in D_i $  and   $\mu \rightarrow  \mu-2$.
                \end{itemize} 
                In each of these cases, the measure $\mu$ drops by at least one. 
    
            \item \textbf{Case 3.} \textbf{The violated triangle $T=xyz$ is such that $xy \in D_i, yz \in M, xz \in M$}.
                In this case, we branch in the following three cases 
                \begin{itemize}
                    \item  $xy \in D_i  \rightarrow xy \in M $ and decrease $\mu \rightarrow  \mu-1$, 
                    \item  $yz \in M  \rightarrow yz \in D_i $  and  $\mu \rightarrow  \mu-1$, 
                    \item  $xz \in M  \rightarrow xz \in D_i $  and   $\mu \rightarrow  \mu-1$.
                \end{itemize} 
                In each of these cases, the measure $\mu$ drops by at least one. 
        \end{itemize}
    \end{mdframed}    
    
    \medskip\noindent\emph{Running Time Analysis.} 
    In our branching algorithm, each branch decreases the measure $\mu$ by at least one. 
    For each violated triangle, we branch into at most three cases, so the branching tree has at most 
    $3^{\mu} \leq 3^{2k}$ leaves. 
    
    \noindent In our branching tree, observe any computational path from the root to the leaf.
    It contains three types of node --- branching node (node with degree at least $3$), leaf node and non-branching nodes (nodes with degree $2$).
    At each leaf, when we obtain a special instance of \Cuvd, $(\Xcal',M',\mu')$, we spend constant time to check whether our instance has overused the budget or not.
    
    \noindent For each branching node in the search tree, detecting a violated triangle with respect to level $i$ is equivalent to finding an induced $P_3$ in the auxiliary graph $G_i$. 
    Since a graph contains a $P_3$ if and only if it has a component that is not a clique, 
    a $P_3$ can be found in linear time($\Ocal(|X|^{2})$). 
    Since we spend at least $1$ budget at each branching node there are at most $2k$ branching nodes in a computational path from root to the leaf and hence we spend $\Ocal(k\cdot|X|^2)$ time on all the branching nodes in a computational path.
    
    \noindent Non-branching nodes are the nodes where we do not branch but decrease the level $i$ by $1$. 
    If we do not have violated triangles with respect to level $i$ we decrease the level to $i-1$ and this means the auxiliary graph $G_i$ is a cluster graph.
    When we go from level $i$ to level $i-1$ then in the auxiliary graph all the edges (pairs of points) having distance $d_{i-1}$ become non-edges.
    It is known that we can check in linear time  $\Ocal(|D_{i-1}|)$ whether after removing these edges we obtain a cluster graph.
    Note that in a computational path once we go down from level $i$ to $i-1$ we do not use the edges in $D_{\geq i}$ and thus for all the non-branching nodes in the  computational path we spend time  $\sum_{i=1}^{L} \Ocal(|D_i|) = \Ocal(|X|^2)$ time. 
    So we have that in a computational path
    \begin{itemize}
        \item For the \emph{leaf} node we spend constant time.
        \item For a \emph{branching} nodes we spend a total of $\Ocal(k \cdot |X|^2)$ time.
        \item For a \emph{non-branching} node we spend a total of $\Ocal(|X|^2)$ time.
    \end{itemize}
    
    \noindent Thus the time required by our algorithm for a computational path is 
    \[
        \Ocal(1) + \Ocal(|X|^2) + \Ocal(k \cdot |X|^2) = \Ocal(k\cdot|X|^2)
    \]  
    Since we can have at most $3^{2k}$ leaves, the number of computational paths is at most $3^{2k}$, and the overall running time of the algorithm is
    $
        \Ocal\big(9^k \cdot k \cdot |X|^{2}\big).
    $ 
    
    This completes the proof.\end{proof}

\fi

\section{Conclusion and open questions}

In this paper, we studied \uvd from the perspective of parameterized complexity, complementing the existing body of work that has primarily focused on approximation algorithms. We established polynomial kernelization results and single-exponential fixed-parameter algorithms, thereby providing a more fine-grained understanding of the tractability of the problem when parameterized by the number of violated distances.  

We conclude by highlighting three natural and fundamental open questions that arise from this work.

\begin{itemize}
    \item \textbf{Polynomial kernel and single-exponential algorithm for \textsc{Ultrametric Violation Distance}.}
    We showed that \textsc{Ultrametric Violation Distance} admits a kernel with $\mathcal{O}(k^2)$ points. 
    Is it possible to obtain a kernel with $\mathcal{O}(k)$ points? Additionally, can the base of the exponent in the current $9^k$ single-exponential algorithm be improved?
    
    \item \textbf{Polynomial kernel and single-exponential algorithm for \textsc{Metric Violation Distance}.}  
    While \uvd admits a polynomial kernel, it remains open whether the more general \textsc{Metric Violation Distance} problem admits a polynomial kernel when parameterized by the number of modified distances. 
    A related question is whether the running time $k^{\mathcal{O}(k)} n^{\mathcal{O}(1)}$ of the algorithm of Fan, Gilbert, Raichel, Sonthalia, and Van~Buskirk~\cite{fan_et_al:LIPIcs.SWAT.2020.25} for this problem is asymptotically optimal.
    
    \item \textbf{Fixed-parameter tractability of tree metric violation distance.}  
    Since every ultrametric is a special case of a tree metric, a natural generalization of \uvd\ is the problem of transforming a distance function into a tree metric by modifying at most $k$ distances. The parameterized complexity of this \emph{Tree Metric Violation Distance} problem is currently unknown.

\end{itemize}

\bibliography{book_pc.bib}

\ifshort

\else

\appendix

\section{Proof of Proposition~\ref{lemma:poly-if-known-edges}}
    \label{app:poly-if-known-edges}
    \begin{proof}
        The  idea is to set the value of $\dist$ to every edge in $A$ equal to its upper-threshold value  and then monotonically decrease some of these values until the ultrametric condition is satisfied. By modifying thresholds,  we preserve the invariant that each edge in $A$ always attains its current upper threshold.  Hence, at any step, either the current distances already form an ultrametric, or we could further  decrease certain edge values.
    
        We start by setting for every edge $xy \in A$,  $\dist(x,y) := \treshold_{u}(x,y)$. After this, we follow the following greedy procedure.
        If $D$ becomes an ultrametric, we stop and return $D$ as the solution. If $D$ is not an ultrametric, then there exist a triangle $T = xyz$ which does not satisfy the ultrametric conditions. 
        By symmetry, we can assume that $\dist(x,y) > \dist(y,z) \geq \dist(x,z)$.
        The weights of the edges $yz$ and $xz$ cannot be increased: if these edges are not in $A$, then we do not change them, and if they are in $A$, their weights are already at their upper thresholds and therefore cannot be increased further.
        Thus, we have to decrease $\dist(x,y)$. If $xy\notin A$ then we cannot modify $\dist(x,y)$ and conclude that $D$ cannot be converted into ultrametric. 
        Otherwise, if $xy\in A$, and in every solution $\dist(x,y)$ should be at most $\dist(y,z)$. Hence, it is safe to decrease $\dist(x,y)$ to $\dist(y,z)$ and set a new threshold $\treshold_{u}(x,y) := \dist(y,z)$.    

        Firstly, in the above procedure only entries from $\Dcal \cup \Treshold_u \cup \Treshold_l$ are used.
        Additionally, the value of any edge $xy \in A$ can be changed at most $|\Dcal \cup \Treshold_u \cup \Treshold_l| = 3\cdot|X|^{2}$ times.       
        Also, we can check whether a given distance matrix is ultrametric in time $\mathcal{O}(|X|^2)$ \cite{KANNAN199626}.
        Thus, the above procedure takes polynomial time.
    \end{proof}

\fi

\end{document}

%% file: Figures/irrelevant-vertex.tex
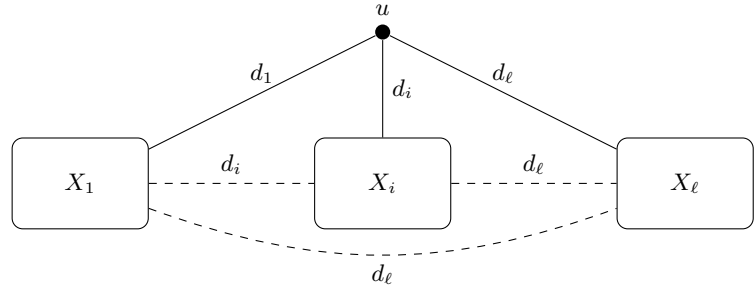
\begin{figure}[t]
\centering
\begin{tikzpicture}[
    scale=1,
    every node/.style={font=\small},
    cluster/.style={draw, rounded corners, minimum width=1.8cm, minimum height=1.2cm},
    >=latex
]

\node[circle,fill,inner sep=2pt,label=above:$u$] (u) at (0,2) {};

\node[cluster] (X1) at (-4,0) {$X_1$};
\node[cluster] (X2) at (0,0) {$X_i$};
\node[cluster] (X3) at (4,0) {$X_\ell$};

\draw (u) -- node[above] {$d_1$} (X1);
\draw (u) -- node[right] {$d_i$} (X2);
\draw (u) -- node[above] {$d_\ell$} (X3);




\draw[dashed] (X1) -- node[above] {$d_i$} (X2);
\draw[dashed] (X2) -- node[above] {$d_\ell$} (X3);
\draw[dashed,bend right=20] (X1) to node[below] {$d_\ell$} (X3);

\end{tikzpicture}
\caption{
    The vertex $u$ (which does not lies in any bad triangles) partitions $X\setminus\{u\}$ into sets $X_1,\ldots,X_\ell$ according to their distance from $u$.
    Distances between different sets are fixed: every edge between $X_i$ and $X_j$ ($i<j$) has length $d_j$.
}
\label{fig:irrelevant-vertex}
\end{figure}